\documentclass[final,1p,times]{elsarticle}
\usepackage[skip=4pt]{parskip}
\usepackage{amsthm, amsmath, amsfonts, amssymb, bm, bbm} % Math packages
\usepackage{nicefrac}       % Compact symbols for 1/2, etc.
\usepackage{mathtools}
\usepackage{algorithm}
\usepackage{algorithmic}
\usepackage{array}
\usepackage{booktabs}
\usepackage{multirow}
\usepackage[labelfont=sf,textfont=sf]{subcaption}
\usepackage{float}
\usepackage{textcomp}
\usepackage{stfloats}
\usepackage{url}
\usepackage{verbatim}
\usepackage{graphicx}
\usepackage{xcolor} 
\usepackage{pgf}
\usepackage{tikz}
\usepackage[colorlinks,bookmarks=false]{hyperref}

\newcommand{\N}{\mathbb{N}} % Naturals
\newcommand{\Z}{\mathbb{Z}} % Integers
\newcommand{\R}{\mathbb{R}} % Reals

\newcommand{\func}[3]{#1:#2\rightarrow#3}

\usepackage{xargs}
\newcommandx{\seq}[3][2=k\in\N,3={}]{(#1)_{#2}^{#3}}

\DeclareMathOperator{\prox}{prox} % Proximal mapping
\DeclareMathOperator*{\minimize}{minimize} % Minimize
\DeclareMathOperator*{\argmin}{arg\,min} % Argmin with subscript alignment
\DeclareMathOperator{\clamp}{clamp} % Clamp
\DeclareMathOperator{\round}{round} % Round

\newtheorem{theorem}{Theorem}

\newtheorem{prop}{Proposition}

\newtheorem{lem}{Lemma}

\newcounter{mysubequations}
\renewcommand{\themysubequations}{(\alph{mysubequations})}
\newcommand{\mysubnumber}{\refstepcounter{mysubequations}\themysubequations}

\newcommand{\remove}[1]{}

\newcommand{\drawrectangleonfigure}[7]{%
    \begin{tikzpicture}
        \node[anchor=south west,inner sep=0] (image) at (0,0) {\includegraphics[width=#2]{#3}};
        \begin{scope}[x={(image.south east)},y={(image.north west)}]
            \draw[#1, very thick] (#4,#5) rectangle (#6,#7);
        \end{scope}
    \end{tikzpicture}
}

\journal{Information Sciences}

\begin{document}

\begin{frontmatter}

    \title{Quantization-aware Matrix Factorization for\\ Low Bit Rate Image Compression} %% Article title

    \author[KUL]{Pooya~Ashtari\corref{cor}\fnref{contr}}%\ead{pooya.ashtari@esat.kuleuven.be}
    \author[KUL]{Pourya~Behmandpoor\fnref{contr}}%\ead{pourya.behmandpoor@esat.kuleuven.be}
    \author[Std]{Fateme~Nateghi~Haredasht}%\ead{fnateghi@stanford.edu}
    \author[Std]{Jonathan~H.~Chen}%\ead{jonc101@stanford.edu}
    \author[KUL]{Panagiotis~Patrinos}%\ead{panos.patrinos@esat.kuleuven.be}
    \author[KUL]{Sabine~Van~Huffel}%\ead{sabine.vanhuffel@esat.kuleuven.be}

    \affiliation[KUL]{organization={Department of Electrical Engineering (ESAT), STADIUS Center, KU Leuven},
                country={Belgium}}

    \affiliation[Std]{organization={Stanford Center for Biomedical Informatics Research, Stanford University},
                city={CA},
                country={USA}}

    \cortext[cor]{Corresponding author: Pooya Ashtari (pooya.ashtari@esat.kuleuven.be).}
    \fntext[contr]{Pooya Ashtari and Pourya Behmandpoor contributed equally to this work.}

    %% Abstract
    \begin{abstract}
    Lossy image compression is essential for efficient transmission and storage. Traditional compression methods mainly rely on discrete cosine transform (DCT) or singular value decomposition (SVD), both of which represent image data in continuous domains and, therefore, necessitate carefully designed quantizers. Notably, these methods consider quantization as a separate step, where quantization errors cannot be incorporated into the compression process. The sensitivity of these methods, especially SVD-based ones, to quantization errors significantly degrades reconstruction quality. To address this issue, we introduce a quantization-aware matrix factorization (QMF) to develop a novel lossy image compression method. QMF provides a low-rank representation of the image data as a product of two smaller factor matrices, with elements constrained to bounded integer values, thereby effectively integrating quantization with low-rank approximation. We propose an efficient, provably convergent iterative algorithm for QMF using a block coordinate descent (BCD) scheme, with subproblems having closed-form solutions. Our experiments on the Kodak and CLIC 2024 datasets demonstrate that our QMF compression method consistently outperforms JPEG at low bit rates below 0.25 bits per pixel (bpp) and remains comparable at higher bit rates. We also assessed our method's capability to preserve visual semantics by evaluating an ImageNet pre-trained classifier on compressed images. Remarkably, our method improved top-1 accuracy by over 5 percentage points compared to JPEG at bit rates under 0.25 bpp. The project is available at \href{https://github.com/pashtari/lrf}{https://github.com/pashtari/lrf}.
\end{abstract}

\end{frontmatter}

\section{Introduction} \label{sec:introduction}

Lossy image compression involves reducing the storage size of digital images by discarding some image data that are redundant or less perceptible to the human eye. This is crucial for efficiently storing and transmitting images, particularly in applications where bandwidth or storage resources are limited, such as web browsing, streaming, and mobile platforms. Lossy image compression methods enable adjusting the degree of compression, providing a selectable tradeoff between storage size and image quality. Widely used methods such as JPEG \cite{wallace1991jpeg} and JPEG 2000 \cite{skodras2001jpeg} follow the \emph{transform coding} paradigm \cite{goyal2001theoretical}. They use orthogonal linear transformations, such as discrete cosine transform (DCT) \cite{ahmed1974discrete} and discrete wavelet transform (DWT) \cite{antonini1992image}, to decorrelate small image blocks. Since these transforms map image data into a continuous domain, quantization is necessary before coding into bytes. Unfortunately, as quantization errors can significantly degrade compression performance, the quantizers must be carefully crafted to minimize this impact, which further complicates codec design.

Another promising paradigm relies on low-rank approximation techniques, with singular value decomposition (SVD) being a notable example. SVD is recognized as the deterministically optimal transform for energy compaction \cite{andrews1976singular}. In practice, current SVD-based methods \cite{andrews1976singular, prasantha2007image, hou2015sparse} can represent image data only with factors that contain floating-point elements, necessitating a quantization step prior to any byte-level processing. The quantization step often introduces errors, which result in suboptimal compression performance.

Motivated by this, we introduce quantization-aware matrix factorization (QMF) and, based on it, develop an effective lossy image compression method. Unlike traditional compression methods, the proposed approach integrates quantization into the optimization process rather than treating it as a separate step before byte-level processing. Our QMF formulation provides a low-rank representation of image data as the product of two smaller factor matrices. The quantization is integrated via introducing constraints in the optimization process, where the elements of the factor matrices are constrained to \emph{bounded integer} values. These elements, with discrete values represented as bounded integers, can be directly stored using standard integral data types---such as \texttt{int8} and \texttt{int16} supported by programming languages---and losslessly processed, making QMF arguably better suited than SVD for image compression. Another advantage of QMF is that the reshaped factor matrices can be treated as 8-bit grayscale images, allowing any lossless image compression standard to be seamlessly integrated into the proposed framework. We propose an efficient iterative algorithm for QMF using a block coordinate descent (BCD) scheme, where each column of a factor matrix is taken as a block and updated one at a time using a closed-form solution.

Our contributions are summarized as follows. We propose a novel optimization framework that enables the integration of quantization and low-rank approximation for image compression. Moreover, we introduce an efficient algorithm for the QMF problem and prove its convergence. Finally, to the best of our knowledge, this work is the first effort to explore QMF for image compression, presenting the first algorithm based on a low-rank approach that significantly outperforms SVD and competes favorably with JPEG, particularly at low bit rates. Our method narrows the gap between factorization and quantization by integrating them into a single layer and optimizing the compression system.
\section{Related Work} \label{sec:related_work}

\paragraph{Transform Coding}
Transform coding is a widely used approach in lossy image compression, leveraging mathematical transforms to decorrelate pixel values and represent image data more compactly. One of the earliest and most influential methods is the discrete cosine transform (DCT) \cite{ahmed1974discrete}, used in JPEG \cite{wallace1991jpeg}, which converts image data into the frequency domain, prioritizing lower frequencies to retain perceptually significant information. The discrete wavelet transform (DWT) \cite{antonini1992image}, used in JPEG 2000 \cite{skodras2001jpeg}, offers improved performance by capturing both frequency and location information, leading to better handling of edges and textures \cite{shapiro1993embedded}. More recently, the WebP \cite{google2011webp} and HEIF \cite{lainema2016hevc, hannuksela2015high} formats combine DCT and intra-frame prediction to achieve superior compression and quality compared to JPEG.

\paragraph{Learned Image Compression (LIC)}
Recently, learned image compression (LIC) has gained attention for potentially outperforming traditional methods by leveraging deep neural networks. \citet{balle2018variational} pioneered this area with an end-to-end trainable convolutional neural network based on variational autoencoders.\remove{Mentzer et al. (2020) ("High-fidelity generative image compression") integrated generative adversarial networks into LIC, developing an effective generative lossy compression system.} \citet{cheng2020learned} incorporated a simplified attention module and discretized Gaussian mixture likelihoods for achieving a more accurate and flexible entropy model. \citet{liu2023learned} combined transformers and CNNs to exploit the local modeling ability of convolutions and the global modeling ability of the attention mechanism. \citet{yang2024lossy} introduced diffusion models into LIC, using a denoising decoder to iteratively reconstruct a compressed image. Despite these advancements, the high computational complexity of LIC methods remains a significant limitation, particularly for real-time applications and resource-constrained environments.

\paragraph{Low-rank Techniques}
Low-rank approximation can provide a compact representation by decomposing image data into smaller components. Notably, truncated singular value decomposition (tSVD) is a classical technique that decomposes images into singular values and vectors, retaining only the most significant components to achieve compression \cite{andrews1976singular, prasantha2007image}. \citet{hou2015sparse} proposed sparse low-rank matrix approximation (SLRMA) for data compression, which is able to explore both the intra- and inter-coherence of data samples simultaneously from the perspective of optimization and transformation.
More recently, \citet{yuan2020image} introduced a graph-based low-rank regularization to reduce compression artifacts near block boundaries at low bit rates.

\paragraph{Integer Matrix Factorization}
There are applications where meaningful representation of data as discrete factor matrices is crucial. While typical low-rank techniques like SVD and nonnegative matrix factorization (NMF) are inappropriate for such applications, integer matrix factorization (IMF) ensures the integrality of factors to achieve this goal. \citet{lin2005integer} investigates IMF to effectively handle discrete data matrices for cluster analysis and pattern discovery. \citet{dong2018integer} introduce an alternative least squares method for IMF, verifying its effectiveness with some data mining applications. However, the application of IMF in image compression remains unexplored.

While existing IMF methods generally constrain factor elements to the entire set of integer values, we propose quantized matrix factorization (QMF), which minimizes the objective function over a bounded interval of integers, thereby modeling a uniformly quantized domain. Furthermore, we introduce a block coordinate descent (BCD)-based algorithm to solve the QMF problem, which is both computationally efficient and provably convergent. This work investigates the potential of QMF for image compression, arguing that it can serve as a powerful tool for this purpose.
\section{Method} \label{sec:method}

\subsection{Overall Encoding Framework} \label{sec:overall_framework}

The proposed compression method follows a \emph{transform coding} paradigm, but it does not involve a separate quantization step. Figure \ref{fig:qmf_encoder} illustrates an overview of our encoding pipeline based on quantization-aware matrix factorization (QMF). The encoder accepts an RGB image with dimensions $H \times W$ and a color depth of 8 bits, represented by the tensor $\bm{\mathcal{X}} \in \{0, \ldots, 255\}^{3 \times H \times W}$. Each step of encoding is described in the following.

\begin{figure*}[!t]
    \centering
    \includegraphics[width=\textwidth]{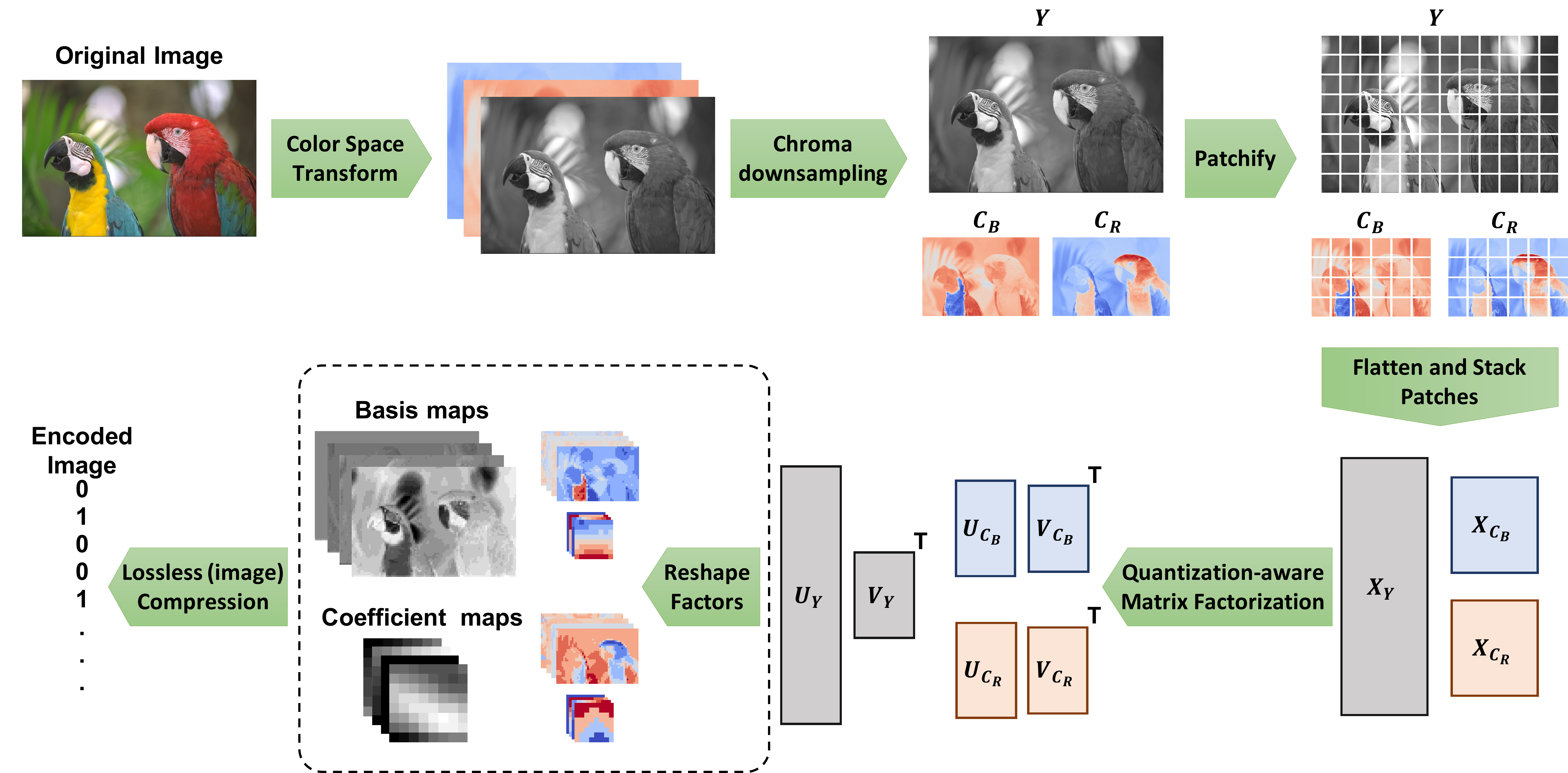}
    \caption{An illustration of the encoder for our image compression method.}
    \label{fig:qmf_encoder}
\end{figure*}

\paragraph{Color Space Transformation}
Analogous to the JPEG standard, the image is initially transformed into the YC\textsubscript{B}C\textsubscript{R} color space. Let $\bm{Y} \in [0, 255]^{H \times W}$ represent the \emph{luma} component, and $\bm{C}_B, \bm{C}_R \in [0, 255]^{\frac{H}{2} \times \frac{W}{2}}$ represent the blue-difference and red-difference \emph{chroma} components, respectively. Note that as a result of this transformation, the elements of the \emph{luma} ($\bm{Y}$) and \emph{chroma} ($\bm{C}_B$, $\bm{C}_R$) matrices are not limited to integers and can take any value within the interval $[0, 255]$.

\paragraph{Chroma Downsampling}
After conversion to the YC\textsubscript{B}C\textsubscript{R} color space, the \emph{chroma} components $\bm{C}_B$ and $\bm{C}_R$ are downsampled using average-pooling with a kernel size of $(2, 2)$ and a stride of $(2, 2)$, similar to the process used in JPEG. This downsampling exploits the fact that the human visual system perceives far more detail in brightness information (\emph{luma}) than in color saturation (\emph{chroma}).

\paragraph{Patchification}
After \emph{chroma} downsampling, we have three components:  the \emph{luma} component $\bm{Y} \in [0, 255]^{H \times W}$ and the \emph{chroma} components $\bm{C}_B, \bm{C}_R \in [0, 255]^{\frac{H}{2} \times \frac{W}{2}}$. Each of the matrices is split into non-overlapping $8 \times 8$ patches. If a dimension of a matrix is not divisible by 8, the matrix is first padded to the nearest size divisible by 8 using reflection of the boundary values. These patches are then flattened into row vectors and stacked vertically to form matrices $\bm{X}_{Y} \in [0, 255]^{\frac{HW}{64} \times 64}$, $\bm{X}_{C_B} \in [0, 255]^{\frac{HW}{256} \times 64}$, and $\bm{X}_{C_R} \in [0, 255]^{\frac{HW}{256} \times 64}$. Later, these matrices will be low-rank approximated using QMF. Note that this patchification technique differs from the block splitting in JPEG, where each block is subject to DCT individually and processed independently. This patchification technique not only captures the locality and spatial dependencies of neighboring pixels but also performs better when combined with the matrix decomposition approach for image compression.

\paragraph{Low-rank Approximation}
We now apply a low-rank approximation to the matrices $\bm{X}_{Y}$, $\bm{X}_{C_B}$, and $\bm{X}_{C_R}$, which is the core of our compression method that provides a lossy compressed representation of these matrices.  The low-rank approximation \cite{eckart1936approximation} aims to approximate a given matrix $ \bm{X} \in \mathbb{R}^{M \times N} $ by
\begin{equation} \label{eq:lra}
    \bm{X} \approx \bm{U} \bm{V}^\mathsf{T} = \sum_{r=1}^{R} U_{:r} {V_{:r}}^\mathsf{T},
\end{equation}
where $\bm{U} \in \mathbb{R}^{M \times R}$ and $\bm{V} \in \mathbb{R}^{N \times R}$ are \emph{factor matrices} (or simply \emph{factors}), $R \leq \min(M,N)$ represents the \emph{rank}, $U_{:r}$ and $V_{:r}$ represent the $r$-th columns of $\bm{U}$ and $\bm{V}$, respectively. We refer to $\bm{U}$ as the \emph{basis matrix} and $\bm{V}$ as the \emph{coefficient matrix}. By selecting a sufficiently small value for $R$, the factor matrices $\bm{U}$ and $\bm{V}$, with a combined total of $(M+N)R$ elements, offer a compact representation of the original matrix $\bm{X}$, which has $MN$ elements, capturing the most significant patterns in the image. Depending on the loss function used to measure the reconstruction error between $\bm{X}$ and the product $\bm{U} \bm{V}^\mathsf{T}$, as well as the constraints on the factor matrices $\bm{U}$ and $\bm{V}$, various formulations and variants have been proposed for different purposes \cite{lee2000algorithms, ding2008convex, lin2005integer}. In Section \ref{sec:qmf}, we introduce and elaborate on our variant, termed quantization-aware matrix factorization (QMF), and argue why it is well-suited and effective for image compression.

\paragraph{Lossless Compression}
QMF yields factor matrices $\bm{U}_{Y} \in \{0, \ldots, 255\}^{\frac{HW}{64} \times R}$ and $\bm{V}_{Y} \in \{~0,~~~~~ \ldots\\, 255\}^{64 \times R}$; $\bm{U}_{C_B} \in \{0, \ldots, 255\}^{\frac{HW}{256} \times R}$ and $\bm{V}_{C_B} \in \{0, \ldots, 255\}^{64 \times R}$; and $\bm{U}_{C_R} \in \{0, \ldots, 255\}^{\frac{HW}{256} \times R}$ and $\bm{V}_{C_R} \in \{0, \ldots, 255\}^{64 \times R}$ that correspond to $\bm{X}_{Y}$, $\bm{X}_{C_B}$, and $\bm{X}_{C_R}$. Since these matrices have elements constrained to integer values (allowing seamless integration of quantization with their optimization process), they can be directly encoded using any standard lossless data compression method, such as zlib \cite{deutsch1996zlib}. In contrast, other lossy image compression methods typically require a separate quantization step, introducing errors that cannot be incorporated or considered during the compression process.

Alternatively, we can first reshape the factor matrices by unfolding their first dimension to obtain $R$-channel 2D spatial maps, referred to as \emph{factor maps} and represented by the following tensors:
\begin{align} \label{eq:reshaped_factors}
    \bm{\mathcal{U}}_{Y}                                                 & \in \{0, \ldots, 255\}^{R \times \frac{H}{8} \times \frac{W}{8}}, \nonumber \\
    \bm{\mathcal{U}}_{C_B}, \bm{\mathcal{U}}_{C_R}                       & \in \{0, \ldots, 255\}^{R \times \frac{H}{16} \times \frac{W}{16}},         \\
    \bm{\mathcal{V}}_{Y}, \bm{\mathcal{V}}_{C_B}, \bm{\mathcal{V}}_{C_R} & \in \{0, \ldots, 255\}^{R \times 8 \times 8}. \nonumber
\end{align}
As each channel of a \emph{factor map} can be treated as an 8-bit grayscale image, we can encode it by any standard lossless image compression method such as PNG. For images with a resolution of $H, W \gg 64$, which are most common nowadays, the \emph{basis maps} ($\bm{\mathcal{U}}$) are significantly larger than the \emph{coefficient maps} ($\bm{\mathcal{V}}$), accounting for the majority of the storage space. Interestingly, in practice, the QMF \emph{basis maps} turn out to be meaningful images, each capturing some visual semantic of the image (see Figure \ref{fig:qmf_components} for an example). Therefore, our QMF approach can effectively leverage the power of existing lossless image compression algorithms, offering a significant advantage over current methods. However, in this work, we take the first approach and use the zlib library \cite{deutsch1996zlib} to encode factor matrices, creating a stand-alone codec that is independent from other image compression methods.

\begin{figure*}[!t]
    \centering
    \begin{minipage}{0.22\textwidth}
        \centering
        \begin{subfigure}{\textwidth}
            \centering
            \includegraphics[width=.80\textwidth]{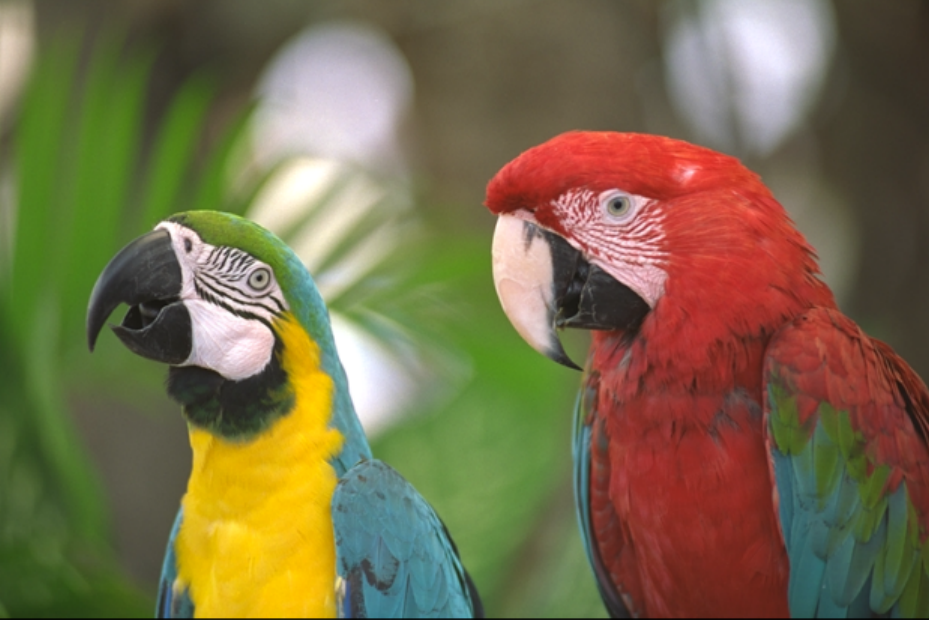}
            \caption{Original}
        \end{subfigure}
    \end{minipage}%
    \begin{minipage}{0.40\textwidth}
        \centering
        \begin{subfigure}{\textwidth}
            \centering
            \includegraphics[width=.95\textwidth]{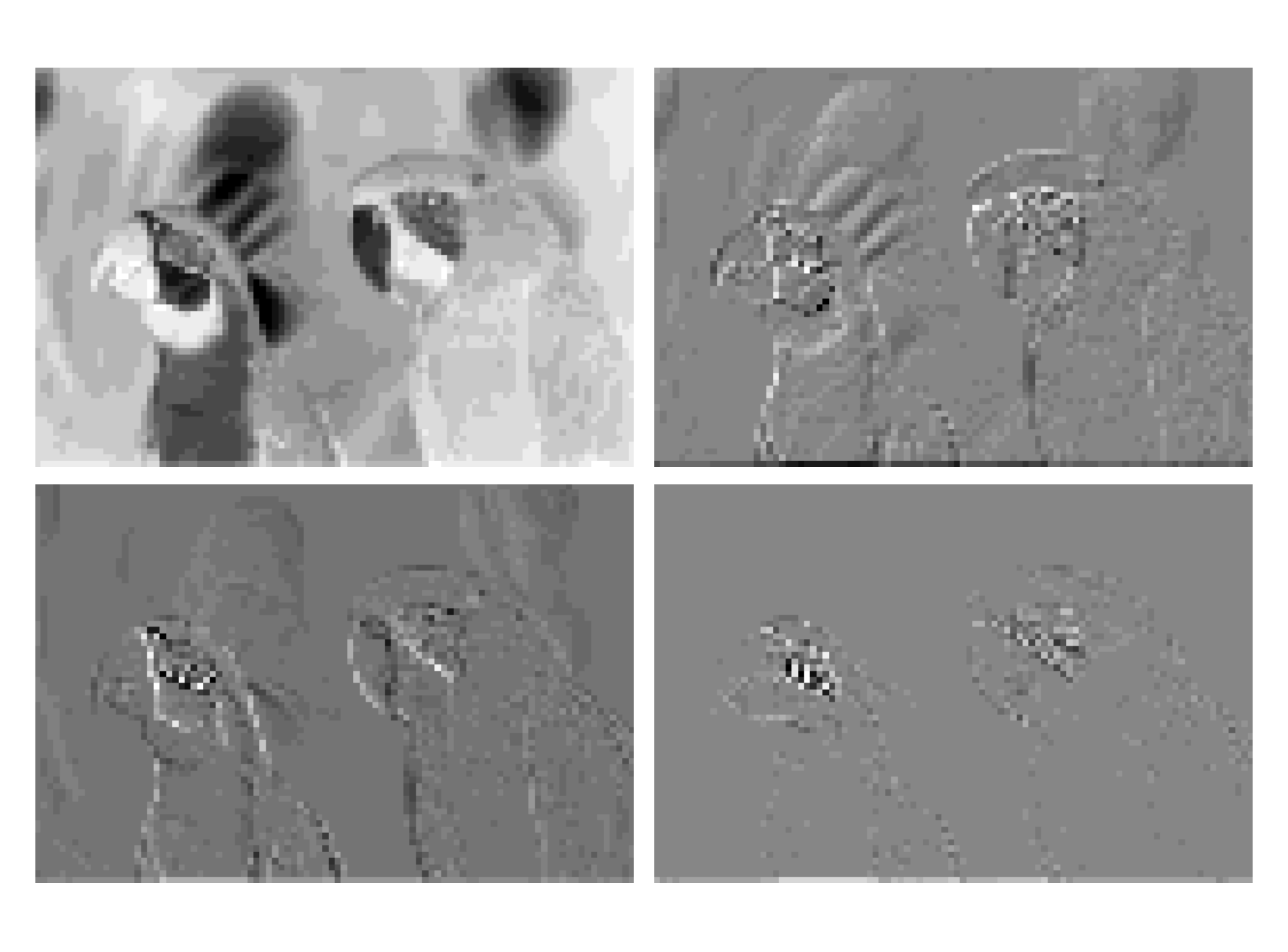}
            \vspace{-10pt}
            \caption{$\bm{\mathcal{U}}_Y$ channels}
        \end{subfigure}
    \end{minipage}%
    \begin{minipage}{0.32\textwidth}
        \centering
        \begin{subfigure}{\textwidth}
            \centering
            \includegraphics[width=.95\textwidth]{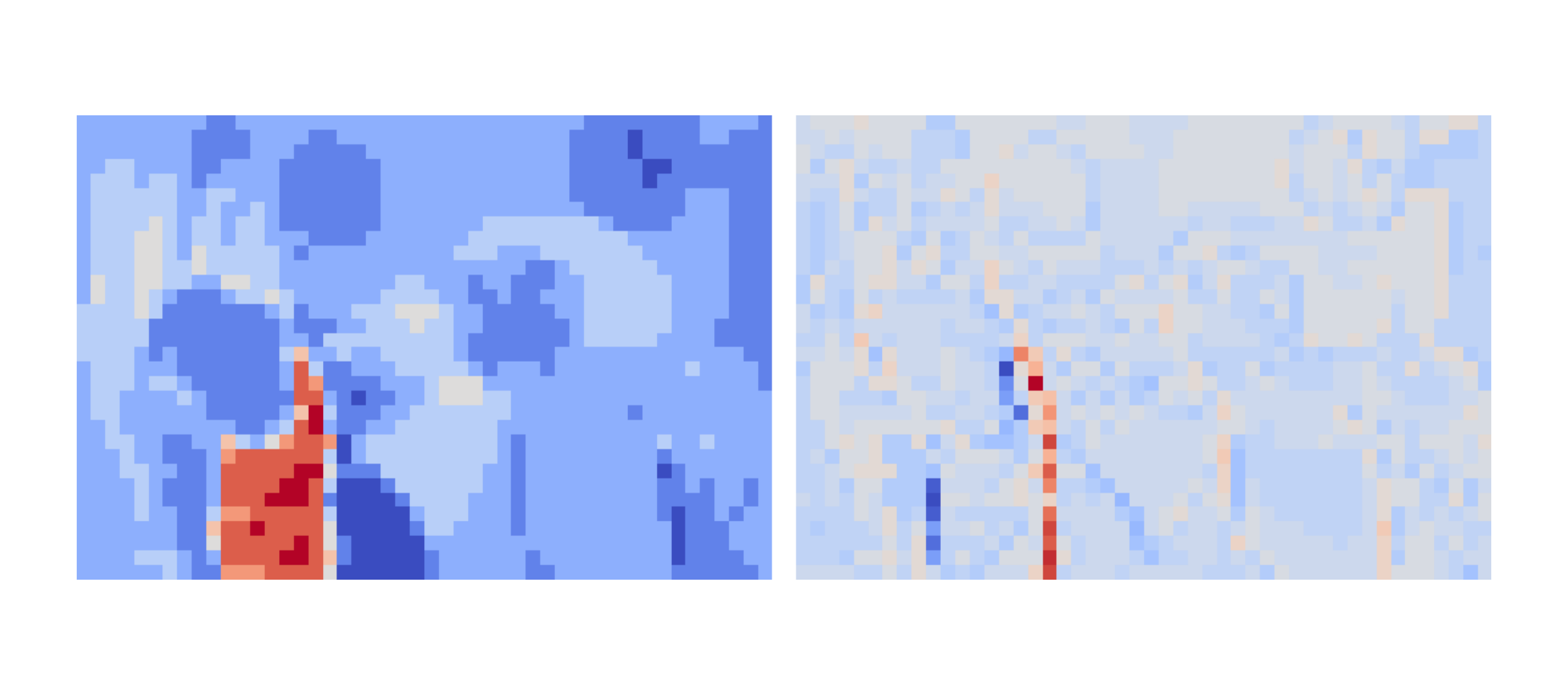}
            \vspace{-10pt}
            \caption{$\bm{\mathcal{U}}_{C_B}$ channels}
        \end{subfigure}
        \begin{subfigure}{\textwidth}
            \centering
            \includegraphics[width=.95\textwidth]{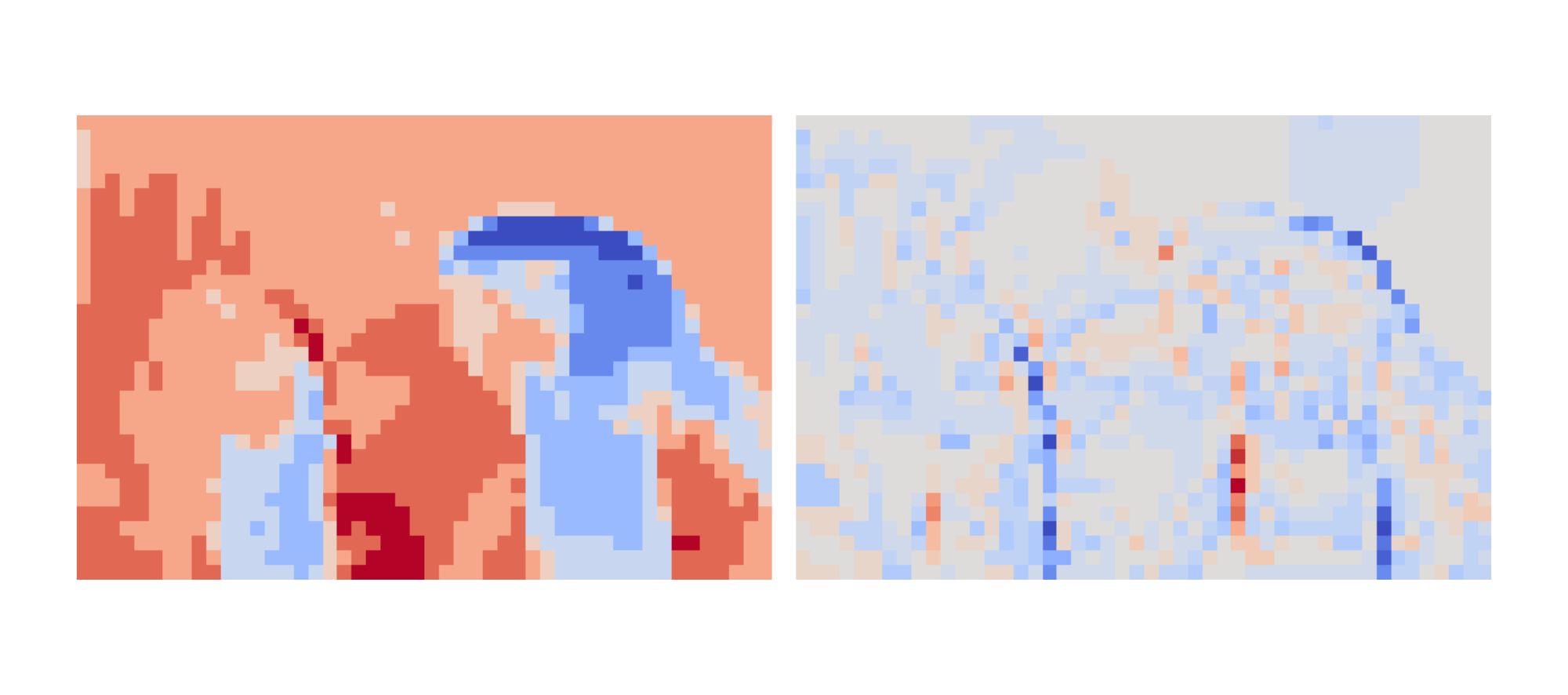}
            \vspace{-10pt}
            \caption{$\bm{\mathcal{U}}_{C_R}$ channels}
        \end{subfigure}
    \end{minipage}
    \caption{The channels of QMF basis maps for the \texttt{kodim23} image from Kodak. (a) shows the original image. The QMF basis maps corresponding to luma (b), blue-difference (c), and red-difference chroma (d) are shown. The channels of basis map with higher energy maintain the overall texture of the original image, where channels with lower energy focus more on subtle changes.}
    \label{fig:qmf_components}
\end{figure*}

\subsection{Decoding} \label{sec:decoding}

The decoder receives an encoded image and reconstructs the RGB image by applying the inverse of the operations used by the encoder, starting from the last layer and moving to the first. Initially, the factor matrices are produced by losslessly decompressing the encoded image. The matrices $\bm{X}_{Y}$, $\bm{X}_{C_B}$, and $\bm{X}_{C_R}$ are calculated through the product of the corresponding factor matrices, according to \eqref{eq:lra}. The \emph{luma} and downsampled \emph{chroma} components are then obtained by reshaping $\bm{X}_{Y}$, $\bm{X}_{C_B}$, and $\bm{X}_{C_R}$ back into their spatial forms, following the inverse of the patchification step. Subsequently, the downsampled \emph{chroma} components are upsampled to their original size using nearest-neighbor interpolation. Finally, the YC\textsubscript{B}C\textsubscript{R} image is converted back into an RGB image.

\subsection{Quantization-aware Matrix Factorization (QMF)} \label{sec:qmf}

The main building block of our method is quantization-aware matrix factorization (QMF), which is responsible for the lossy compression of matrices obtained through patchification. QMF can be framed as an optimization problem, aiming to minimize the reconstruction error between the original matrix $\bm{X} \in \mathbb{R}^{M \times N}$ and the product $\bm{U} \bm{V}^\mathsf{T}$, while ensuring, as an integrated quantization step, that the elements of the factor matrices $\bm{U}$ and $\bm{V}$ are integers within a specified interval $[\alpha,\beta]$ with integer endpoints, i.e., $\alpha,\beta \in \Z$. Formally, the QMF problem can be expressed as:
\begin{align} \label{eq:qmf_problem}
    \minimize_{\bm{U}, \bm{V}}  \quad & \ \  \| \bm{X} - \bm{U} \bm{V}^\mathsf{T} \|_\text{F}^2 \nonumber                                             \\
    \text{subject to} \quad                 & \ \ \bm{U} \in \mathbb{Z}_{[\alpha,\beta]}^{M \times R}, \bm{V} \in \mathbb{Z}_{[\alpha,\beta]}^{N \times R},
\end{align}
where $\|\cdot\|_\text{F}$ denotes the Frobenius norm; $R \leq \min(M,N)$ represents the \emph{rank}; and $\mathbb{Z}_{[\alpha,\beta]} \triangleq [\alpha,\beta] \cap \mathbb{Z}$ denotes the set of integers within $[\alpha,\beta]$. Without constraints on the factors, the problem would have an analytic solution through singular value decomposition (SVD), as addressed by the Eckart–Young–Mirsky theorem \cite{eckart1936approximation}. If only a nonnegativity constraint were applied (without integrality), variations of nonnegative matrix factorization (NMF) would emerge \cite{lee2000algorithms, gillis2020nonnegative}. The QMF problem \eqref{eq:qmf_problem} poses a challenging integer program, with finding its global minima known to be NP-hard \cite{dong2018integer, van1981another}. Only a few iterative algorithms \cite{dong2018integer, lin2005integer} have been proposed to find a ``good solution'' for some QMF variants in contexts other than image compression. In Section \ref{sec:bcd}, we propose an efficient iterative algorithm for the QMF problem \eqref{eq:qmf_problem}.

\begin{algorithm}[!t]
    \caption{The proposed block coordinate descent (BCD) algorithm for QMF. \label{alg:bcd_for_qmf}}
    \begin{algorithmic}[1]
        \REQUIRE $\bm{X} \in \mathbb{R}^{M \times N}$, factorization rank $R$, factor bounds $[\alpha,\beta]$, \# iterations $K$
        \ENSURE Factor matrices $\bm{U} \in \mathbb{Z}_{[\alpha,\beta]}^{M \times R}$ and $\bm{V} \in \mathbb{Z}_{[\alpha,\beta]}^{N \times R}$
        \STATE Initialize $\bm{U}^{\rm init}$, $\bm{V}^{\rm init}$ using the truncated SVD method, provided by \eqref{eq:initialization_u} and \eqref{eq:initialization_v}, and set $k=0$
        \WHILE{$k < K$}
        \STATE $k \gets k+1$
        \STATE $\bm{A} \gets \bm{X} \bm{V}^k$
        \STATE $\bm{B} \gets \bm{V}^{k^{\mathsf{T}}} \bm{V}^k$
        \FOR{$r = 1, \ldots, R$}
        \STATE \label{alg:step:u_update:1} $\displaystyle U_{:r}^{k+1/2} \gets \frac{A_{:r} - \sum_{s = 1}^{r-1} B_{sr} U_{:s}^{k+1} - \sum_{s=r+1}^M B_{sr} U_{:s}^k}{\| V_{:r}^k \|^2}$
        \vspace{2pt}
        \STATE \label{alg:step:u_update:2} $\displaystyle U_{:r}^{k+1} \gets \clamp_{[\alpha,\beta]}(\round( U_{:r}^{k+1/2}))$
        \ENDFOR
        \STATE $\bm{A} \gets \bm{X}^\mathsf{T} \bm{U}^{k+1}$
        \STATE $\bm{B} \gets \bm{U}^{{k+1}^\mathsf{T}} \bm{U}^{k+1}$
        \FOR{$r = 1, \ldots, R$}
        \STATE \label{alg:step:v_update:1} $\displaystyle V_{:r}^{k+1/2} \gets \frac{A_{:r} - \sum_{s = 1}^{r-1} B_{sr} V_{:s}^{k+1} - \sum_{s = r+1}^N B_{sr} V_{:s}^k}{\| U_{:r}^{k+1} \|^2}$
        \vspace{2pt}
        \STATE \label{alg:step:v_update:2} $\displaystyle V_{:r}^{k+1} \gets \clamp_{[\alpha,\beta]}(\round(V_{:r}^{k+1/2}))$
        \ENDFOR
        \ENDWHILE
        \RETURN $(\bm{U}^K, \bm{V}^K)$
    \end{algorithmic}
\end{algorithm}

The existing lossy image compression methods based on SVD and NMF approach the problem as an optimization task, followed by a separate quantization step. The optimization focuses on finding factors that minimize the reconstruction error. Before any byte-level processing, a quantization step is applied to project the floating-point elements of the resulting factors onto a set of discrete values. However, because the quantization is performed separately from the optimization, the quantization errors cannot be incorporated into the compression process. This separation leads to suboptimal compression performance (as demonstrated in Section \ref{sec:experiments}) and additional complications in designing quantizers.
In contrast, our QMF formulation, through the \emph{constrained} optimization in \eqref{eq:qmf_problem}, produces integer factor matrices that minimize the reconstruction error while ensuring that all elements are discrete. These integer factor matrices can be directly stored and processed losslessly without introducing roundoff errors. The reason for limiting the feasible region to $[\alpha,\beta]$ in our QMF formulation is to enable more compact storage of the factors using standard integral data types, such as \texttt{int8} and \texttt{int16}, supported by programming languages. Given that the elements of the input matrix $\bm{X}$ are in $[0, 255]$, we found the signed \texttt{int8} type, which represents integers from -128 to 127, suitable for image compression applications. As a result, our QMF formulation is well-suited for image compression, effectively integrating the factorization and quantization steps into a single, efficient compression process.

\subsection{Block Coordinate Descent Scheme for QMF} \label{sec:bcd}

We propose an efficient algorithm for QMF using the block coordinate descent (BCD) scheme (aka alternating optimization).
The pseudocode is provided in Algorithm \ref{alg:bcd_for_qmf}.
Starting with some initial parameter values, this approach involves sequentially minimizing the cost function with respect to a single column of a factor at a time, while keeping the other columns of that factor and the entire other factor fixed. This process is repeated until a stopping criterion is met, such as when the change in the cost function value falls below a predefined threshold or the maximum number of iterations is reached. Formally, this involves solving one of the following subproblems at a time:
\begin{align}
    \bm{u}_r \gets \underset{\bm{u}_r \in \mathbb{Z}_{[\alpha,\beta]}^{M \times 1}}{\argmin} \ \lVert \bm{E}_r - \bm{u}_r \bm{v}_r^\mathsf{T} \rVert_\text{F}^2, \label{eq:bcd_subproblem_u} \\
    \bm{v}_r \gets \underset{\bm{v}_r \in \mathbb{Z}_{[\alpha,\beta]}^{N \times 1}}{\argmin} \ \lVert \bm{E}_r - \bm{u}_r \bm{v}_r^\mathsf{T} \rVert_\text{F}^2, \label{eq:bcd_subproblem_v}
\end{align}
where $\bm{u}_r \triangleq U_{:r}$ and $\bm{v}_r \triangleq V_{:r}$ represent the $r$-th columns of $\bm{U}$ and $\bm{V}$, respectively. $\bm{E}_r \triangleq \bm{X} - \sum_{s \neq r}^{R} \bm{u}_s \bm{v}_s^\mathsf{T}$ is the residual matrix. We define one iteration of BCD as a complete cycle of updates across all the columns of both factors. In fact, the proposed algorithm is a $2R$-block coordinate descent procedure, where at each iteration, first the columns of $\bm{U}$ and then the columns of $\bm{V}$ are updated (see Algorithm \ref{alg:bcd_for_qmf}). Note that subproblem \eqref{eq:bcd_subproblem_v} can be transformed into the same form as \eqref{eq:bcd_subproblem_u} by simply transposing its error term inside the Frobenius norm. Therefore, we only need to find the best rank-1 approximation with integer elements constrained within a specific interval. Fortunately, this problem has a closed-form solution, as addressed by Theorem \ref{the:bcd_subproblem} below.

\begin{theorem}[Monotonicity] \label{the:bcd_subproblem}
    The global optima of subproblems \eqref{eq:bcd_subproblem_u} and \eqref{eq:bcd_subproblem_v} can be represented by closed-form solutions as follows:
    \begin{align}
        \textstyle \bm{u}_r \gets \clamp_{[\alpha,\beta]}\Big(\round\Big(\frac{\bm{E}_r \bm{v}_r}{\lVert \bm{v}_r \rVert^2}\Big)\Big), \label{eq:bcd_closed_form_subproblem_u} \\
        \textstyle \bm{v}_r \gets \clamp_{[\alpha,\beta]}\Big(\round\Big(\frac{\bm{E}_r^\mathsf{T} \bm{u}_r}{\lVert \bm{u}_r \rVert^2}\Big)\Big),
        \label{eq:bcd_closed_form_subproblem_v}
    \end{align}
    where $\round(\bm{Z})$ denotes an element-wise operator that rounds each element of $\bm{Z}$ to the nearest integer, and $\clamp_{[\alpha,\beta]}(\bm{Z}) \triangleq \max(\alpha, \min(\bm{Z}, \beta))$ denotes an element-wise operator that clamps each element of $\bm{Z}$ to the interval $[\alpha,\beta]$. Moreover, the cost function in \eqref{eq:qmf_problem} is monotonically nonincreasing over BCD iterations of Algorithm \ref{alg:bcd_for_qmf}, which involve sequential updates of \eqref{eq:bcd_closed_form_subproblem_u} and \eqref{eq:bcd_closed_form_subproblem_v} over columns of $\bm U$ and $\bm V$.
\end{theorem}
\begin{proof}
    See \ref{app:monotonicity_proof} for the proof.
\end{proof}

It is noteworthy that the combination of $\round(\cdot)$ and $\clamp_{[\alpha,\beta]}(\cdot)$ in \eqref{eq:bcd_closed_form_subproblem_u} and \eqref{eq:bcd_closed_form_subproblem_v} can be interpreted as the element-wise projector to $\mathbb{Z}_{[\alpha,\beta]}$. In addition, updates \eqref{eq:bcd_closed_form_subproblem_u} and \eqref{eq:bcd_closed_form_subproblem_v} are presented in Algorithm \ref{alg:bcd_for_qmf} at steps \ref{alg:step:u_update:1} and \ref{alg:step:u_update:2}, and steps \ref{alg:step:v_update:1} and \ref{alg:step:v_update:2}, respectively. In Theorem \ref{thm:convergence}, the convergence of Algorithm \ref{alg:bcd_for_qmf} employing these closed-form solutions is established.

\begin{theorem}[Convergence]\label{thm:convergence}
    Let $\seq{\displaystyle U_{:r}^{k}}[k\in\N]$ and $\seq{\displaystyle V_{:r}^{k}}[k\in\N]$ for $r\in \{1,\dots,R\}$ be sequences generated by the proposed Algorithm \ref{alg:bcd_for_qmf}. Then all sequences are bounded and convergent to a locally optimal point of optimization problem \eqref{eq:qmf_problem}.
\end{theorem}
\begin{proof}
    See \ref{app:convergence_proof} for the proof.
\end{proof}

% \begin{theorem}[Monotonicity] \label{the:bcd_subproblem}
% 	Let $\bm{E} \in \mathbb{R}^{M \times N}$, $\bm{u} \in \mathbb{Z}_{[\alpha,\beta]}^{M \times 1}$, and $\bm{v} \in \mathbb{R}^{N \times 1}$. A global solution to the problem
% 	\begin{equation} \label{eq:bcd_subproblem}
% 		\bm{u}^* \in \underset{\bm{u} \in \mathbb{Z}_{[\alpha,\beta]}^{M \times 1}}{\argmin} \ \lVert \bm{E} - \bm{u} \bm{v}^\mathsf{T} \rVert_\text{F}^2
% 	\end{equation}
% 	is given by
% 	\begin{equation} \label{eq:bcd_subproblem_solution}
% 		\bm{u}^* = \clamp_{[\alpha,\beta]}\Big(\round\Big(\frac{\bm{E} \bm{v}}{\lVert \bm{v} \rVert^2}\Big)\Big),
% 	\end{equation} 
% 	where $\round(\bm{Z})$ denotes the operator that rounds each element of $\bm{Z}$ to the nearest integer, and $\clamp_{[\alpha,\beta]}(\bm{Z}) \triangleq \max(\alpha, \min(\bm{Z}, \beta))$ denotes the operator that clamps each element of $\bm{Z}$ to the interval $[\alpha,\beta]$.
% \end{theorem}

% Using Theorem \ref{the:bcd_subproblem}, we obtain the following update formulas for our BCD algorithm:
% Thanks to Theorem \ref{the:bcd_subproblem}, the proposed algorithm offers an encouraging property that guarantees the cost function is monotonically non-increasing with each update. 

\paragraph*{Initialization}
The initial values of factors can significantly impact the convergence performance of the BCD algorithm. We found that the convergence with naive random initialization can be too slow. To address this issue, we propose an initialization method using SVD. The procedure is straightforward. First, the truncated SVD of the input matrix $\bm{X} \in \mathbb{R}^{M \times N}$ is computed as $\tilde{\bm{U}} \bm{\Sigma} \tilde{\bm{V}}^\mathsf{T}$, where $\bm{\Sigma} \in \mathbb{R}^{R \times R}$ is a diagonal matrix corresponding to the $R$ largest singular values. $\tilde{\bm{U}} \in \mathbb{R}^{M \times R}$ and $\tilde{\bm{V}} \in \mathbb{R}^{N \times R}$ contain the corresponding left-singular vectors and right-singular vectors in their columns, respectively. The initial factors are then calculated as follows:
\begin{align}
    \bm{U}^{\rm init} = \clamp_{[\alpha,\beta]}(\round(\tilde{\bm{U}} \bm{\Sigma}^\frac{1}{2})), \label{eq:initialization_u} \\
    \bm{V}^{\rm init} = \clamp_{[\alpha,\beta]}(\round(\bm{\Sigma}^\frac{1}{2} \tilde{\bm{V}})). \label{eq:initialization_v}
\end{align}
Essentially, this means that instead of performing a constrained optimization, we first low-rank approximate $\bm{X}$ and then satisfy the constraints by projecting the elements of the resulting factor matrices onto $\mathbb{Z}_{[\alpha,\beta]}$.

\section{Experiments} \label{sec:experiments}

\subsection{Setup} \label{sec:setup}

\paragraph{QMF Configuration}
In our QMF implementation, we used a default patch size of $8 \times 8$. The default factor bounds were set to $[-16, 15]$. Unless otherwise specified, the number of BCD iterations was set to 10 although our ablation studies in Section \ref{sec:ablation_studies} suggest that even 2 iterations may suffice in practice (see Figure \ref{fig:iters_ablation_psnr}). For lossless compression of factors, we encoded and decoded each column of a factor separately using the zlib library \cite{deutsch1996zlib}. We also tested other lossless compression methods, such as zstd \cite{zstandard} and Huffman coding \cite{huffman1952method}, which demonstrated a comparable performance. However, as zlib is well-established, simple, and offers fast performance in Python, the experimental results are reported for this compression method.

\paragraph{Baseline Codecs}
We compared our QMF method against JPEG and SVD baselines. For JPEG compression, we used the Pillow library \cite{clark2015pillow}. Our SVD baseline follows the same framework as the proposed method (described in \ref{sec:overall_framework}) but substitutes truncated SVD for QMF. This is followed by uniform quantization of the SVD factor matrices before lossless compression using zlib \cite{deutsch1996zlib}. This differs from QMF compression, which benefits from the integrality of factors by directly encoding them with zlib, eliminating the need for a separate quantization step.

\paragraph{Datasets}
To validate the effectiveness of our method, we conducted experiments using the widely-used \textbf{Kodak} dataset \cite{kodak1993}, consisting of 24 lossless images with a resolution of $768 \times 512$. To evaluate the robustness of our method in a higher-resolution setting, we also experimented with the \textbf{CLIC 2024} validation dataset \cite{clic2024}, which contains 30 high-resolution, high-quality images. Additionally, we assessed the compression methods by their ability to retain visual semantics. This was achieved by evaluating a pre-trained ImageNet classifier on compressed images from the \textbf{ImageNet} validation set \cite{russakovsky2015imagenet}, consisting of 50,000 images with a resolution of $224 \times 224$ across 1,000 classes.

\paragraph{Metrics}
To evaluate the rate-distortion performance of methods on the Kodak and CLIC 2024 datasets, we measured the bit rate in bits per pixel (bpp) and assessed the quality of the reconstructed images using peak signal-to-noise ratio (PSNR) and structural similarity index measure (SSIM). Then, these metrics were plotted as functions of bit rate for each method to illustrate their rate-distortion performance. To control the quality of the reconstructed images in QMF and SVD, similar to JPEG, we defined a quality factor $Q\in[0,1]$, where 0 represents the highest compression and 1 represents no compression. To determine the factorization rank $R$ in Algorithm \ref{alg:bcd_for_qmf}, we used $R = \max\left\{\round\left(Q \times \min\{M,N\}\right), 1\right\}$.

More precisely, to construct a rate-distortion curve for each method on each dataset, we evaluated various qualities $Q$ for each image. For each quality, we first measured the PSNR/SSIM values at the corresponding bit rate. Next, PSNR/SSIM values were interpolated at evenly spaced bit rates ranging from 0.05 bpp to 0.5 bpp using LOESS (locally estimated scatterplot smoothing) \cite{cleveland1988locally}. Finally, the interpolated values were averaged over all images at each of these bit rates.

\subsection{Rate-Distortion Performance} \label{sec:rate_distortion_performance}

Figure \ref{fig:rate_distortion} illustrates the rate-distortion curves comparing the performance of QMF, SVD, and JPEG compression methods.

\begin{figure*}[!t]
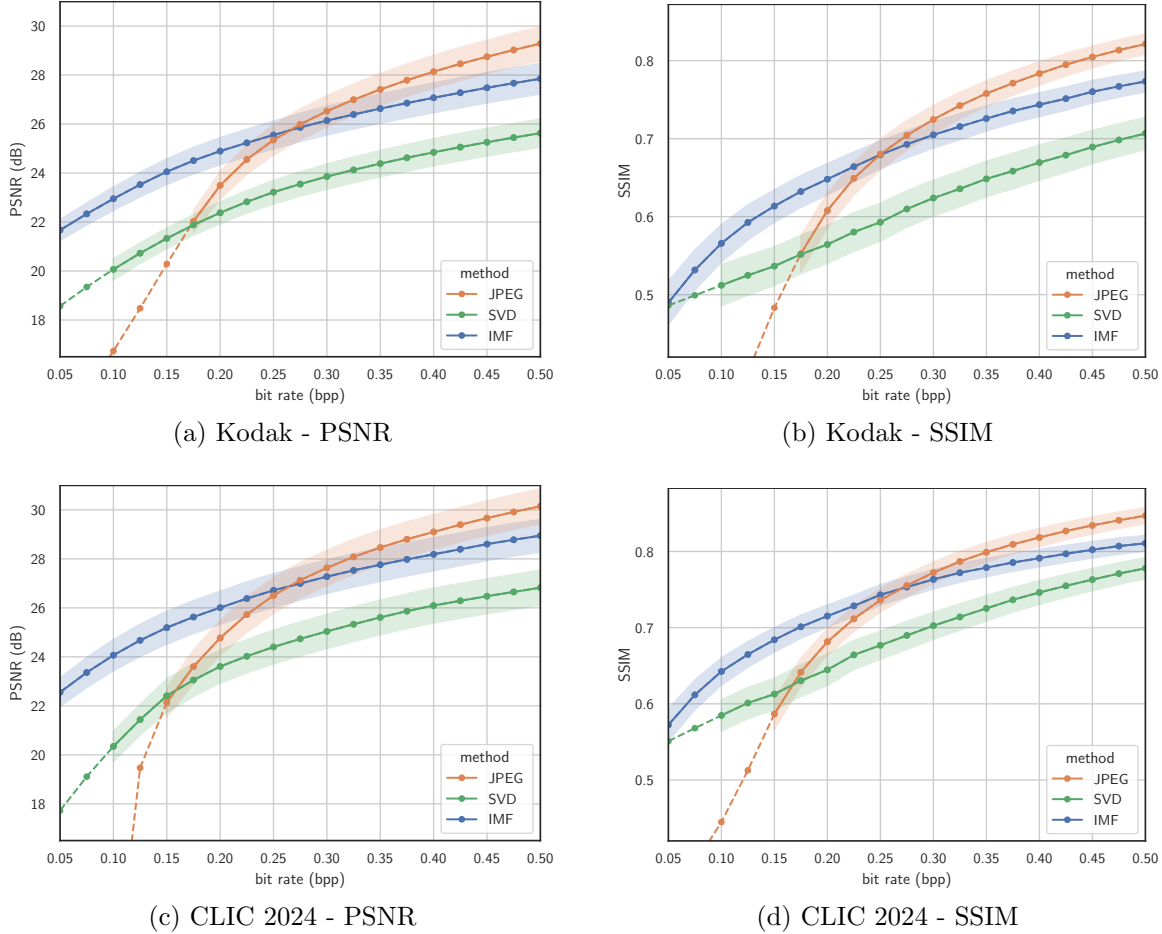

    \centering
    \begin{subfigure}[t]{0.5\textwidth}
        \centering
        \resizebox{.9\textwidth}{!}{\input{figures/kodak_psnr.pgf}}
        \caption{Kodak - PSNR}
        \label{fig:kodak_psnr}
    \end{subfigure}%
    \begin{subfigure}[t]{0.5\textwidth}
        \centering
        \resizebox{.9\textwidth}{!}{\input{figures/kodak_ssim.pgf}}
        \caption{Kodak - SSIM}
        \label{fig:kodak_ssim}
    \end{subfigure}

    \bigskip

    \begin{subfigure}[t]{.5\textwidth}
        \centering
        \resizebox{.9\textwidth}{!}{\input{figures/clic2024_psnr.pgf}}
        \caption{CLIC 2024 - PSNR}
        \label{fig:clic_psnr}
    \end{subfigure}%
    \begin{subfigure}[t]{.5\textwidth}
        \centering
        \resizebox{.9\textwidth}{!}{\input{figures/clic2024_ssim.pgf}}
        \caption{CLIC 2024 - SSIM}
        \label{fig:clic_ssim}
    \end{subfigure}
    \caption{Rate-distortion performance on the Kodak (top panels) and CLIC 2024 (bottom panels) datasets. The average PSNR (left panels) and average SSIM (right panels) for each method are plotted as functions of bit rate. Shaded areas represent standard errors. Dashed lines indicate extrapolated values predicted using LOESS \cite{cleveland1988locally} for extremely low bit rates that are otherwise unattainable.}
    \label{fig:rate_distortion}
\end{figure*}

\paragraph{Kodak}
On the Kodak dataset, as shown in Figures \ref{fig:kodak_psnr} and \ref{fig:kodak_ssim}, our QMF method consistently outperforms JPEG at low bit rates below 0.25 bpp and remains comparable at higher bit rates in terms of both PSNR and SSIM. Furthermore, QMF significantly surpasses the SVD-based baseline across all bit rates. \remove{This superior performance can be attributed to the high sensitivity of SVD to quantization errors during encoding and decoding, which leads to increased reconstruction errors. In contrast, the quantization-free nature of QMF allows for more accurate reconstruction.}

\paragraph{CLIC 2024}
A similar trend is observed with the CLIC 2024 dataset, as shown in Figures \ref{fig:clic_psnr} and \ref{fig:clic_ssim}. Here, the PSNR (Figure \ref{fig:clic_psnr}) and SSIM (Figure \ref{fig:clic_ssim}) results further confirm the competitive performance of QMF across all bit rates, with a particularly notable margin at bit rates lower than 0.25 bpp. Specifically, at a bit rate of 0.15 bpp, QMF achieves an PSNR of over 25 dB, compared to approximately 22 dB for both JPEG and SVD. This supports the robustness of QMF in preserving visual quality across different datasets.

\subsection{Qualitative Performance}

Figure \ref{fig:qualitative_comparison} compares various compression methods using images from the Kodak (top two rows) and CLIC 2024 (bottom two rows) datasets, compressed at similar bit rates.

\begin{figure*}[!t]
    \captionsetup[subfigure]{labelformat=empty,aboveskip=1pt,belowskip=6pt}
    \centering

    \begin{subfigure}[t]{0.25\textwidth}
        \centering
        \textbf{Original Image}
    \end{subfigure}%
    \begin{subfigure}[t]{0.25\textwidth}
        \centering
        \textbf{JPEG}
    \end{subfigure}%
    \begin{subfigure}[t]{0.25\textwidth}
        \centering
        \textbf{SVD}
    \end{subfigure}%
    \begin{subfigure}[t]{0.25\textwidth}
        \centering
        \textbf{QMF}
    \end{subfigure}

    \begin{subfigure}[t]{0.25\textwidth}
        \centering
        \drawrectangleonfigure{blue}{.95\textwidth}{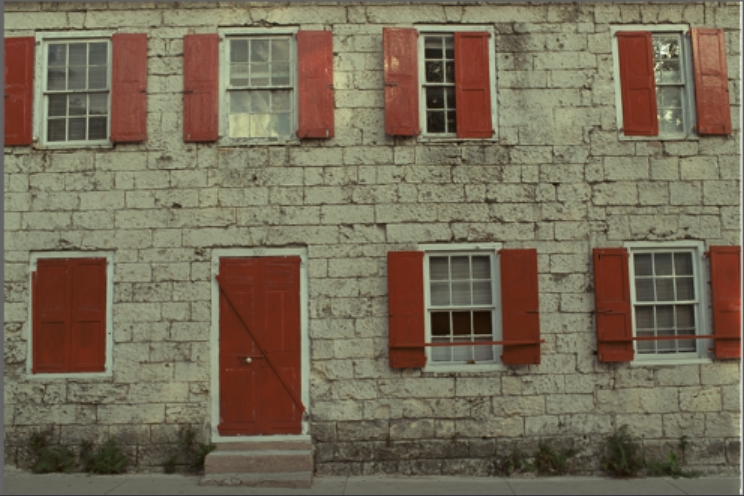}{0.37}{0.62}{0.59}{0.85}
    \end{subfigure}%
    \begin{subfigure}[t]{0.25\textwidth}
        \centering
        \drawrectangleonfigure{blue}{.95\textwidth}{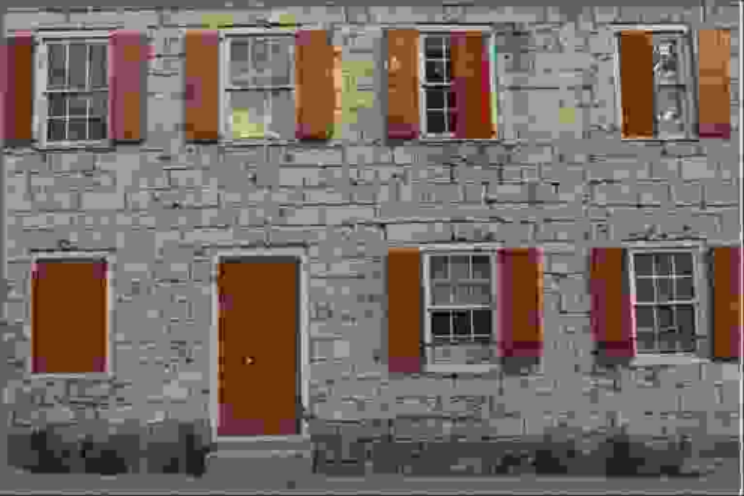}{0.37}{0.62}{0.59}{0.85}
        \caption{\scriptsize\textbf{bpp: 0.21, PSNR: 20.22dB}}
    \end{subfigure}%
    \begin{subfigure}[t]{0.25\textwidth}
        \centering
        \drawrectangleonfigure{blue}{.95\textwidth}{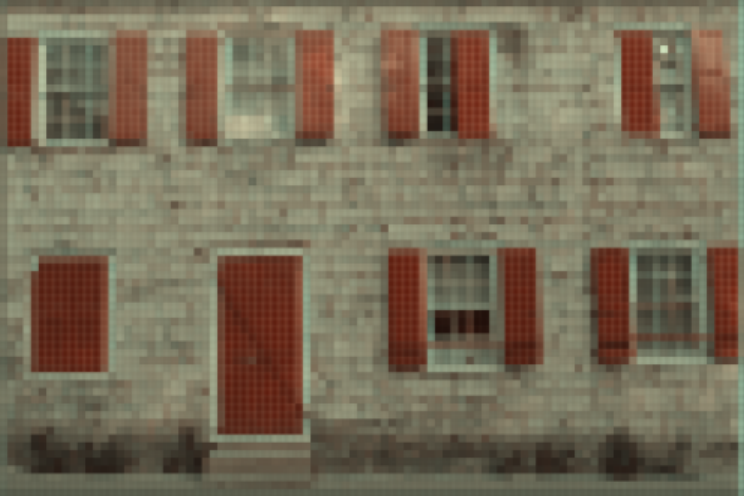}{0.37}{0.62}{0.59}{0.85}
        \caption{\scriptsize\textbf{bpp: 0.22, PSNR: 20.24dB}}
    \end{subfigure}%
    \begin{subfigure}[t]{0.25\textwidth}
        \centering
        \drawrectangleonfigure{blue}{.95\textwidth}{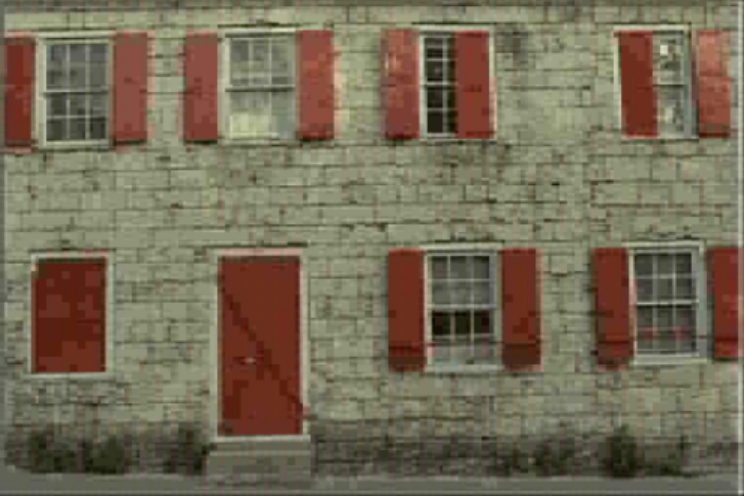}{0.37}{0.62}{0.59}{0.85}
        \caption{\scriptsize\textbf{bpp: 0.21, PSNR: 21.93dB}}
    \end{subfigure}

    \begin{subfigure}[t]{0.25\textwidth}
        \centering
        \drawrectangleonfigure{purple}{.95\textwidth}{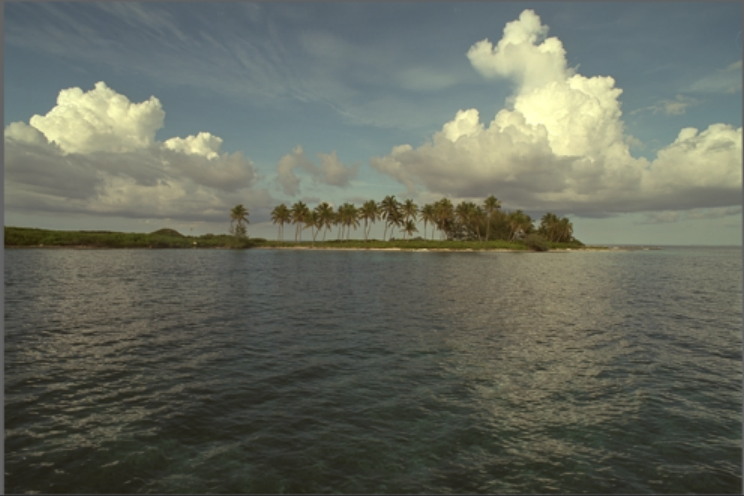}{0.42}{0.53}{0.7}{0.83}
    \end{subfigure}%
    \begin{subfigure}[t]{0.25\textwidth}
        \centering
        \drawrectangleonfigure{purple}{.95\textwidth}{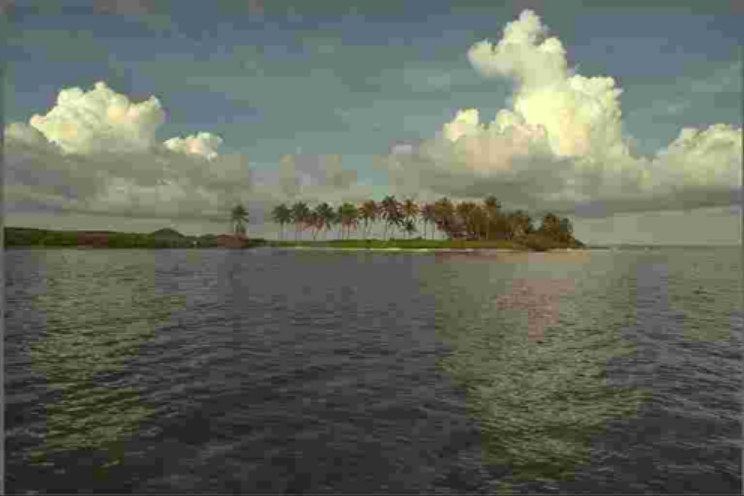}{0.42}{0.53}{0.7}{0.83}
        \caption{\scriptsize\textbf{bpp: 0.27, PSNR: 27.69dB}}
    \end{subfigure}%
    \begin{subfigure}[t]{0.25\textwidth}
        \centering
        \drawrectangleonfigure{purple}{.95\textwidth}{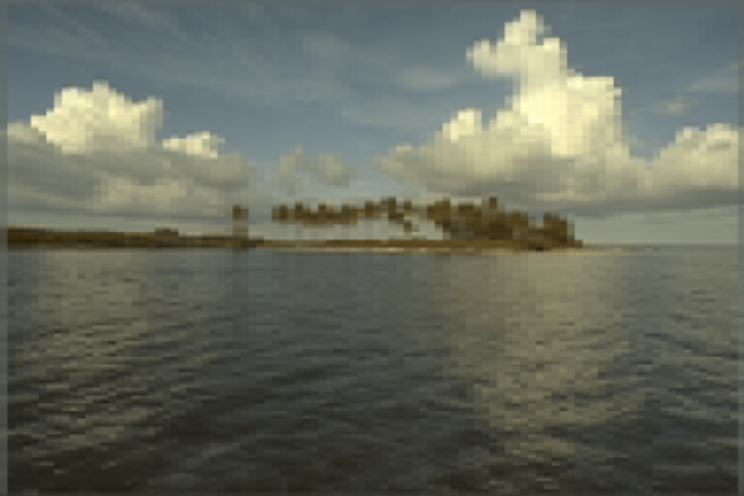}{0.42}{0.53}{0.7}{0.83}
        \caption{\scriptsize\textbf{bpp: 0.27, PSNR: 26.67dB}}
    \end{subfigure}%
    \begin{subfigure}[t]{0.25\textwidth}
        \centering
        \drawrectangleonfigure{purple}{.95\textwidth}{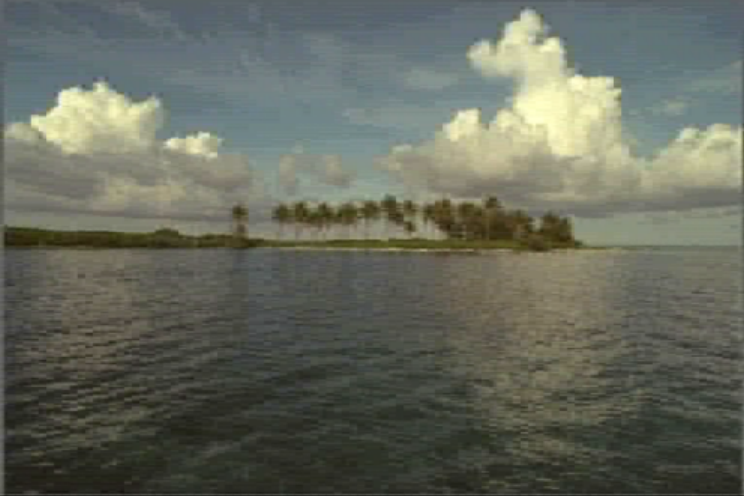}{0.42}{0.53}{0.7}{0.83}
        \caption{\scriptsize\textbf{bpp: 0.26, PSNR: 28.37dB}}
    \end{subfigure}

    \begin{subfigure}[t]{.25\textwidth}
        \centering
        \drawrectangleonfigure{cyan}{.95\textwidth}{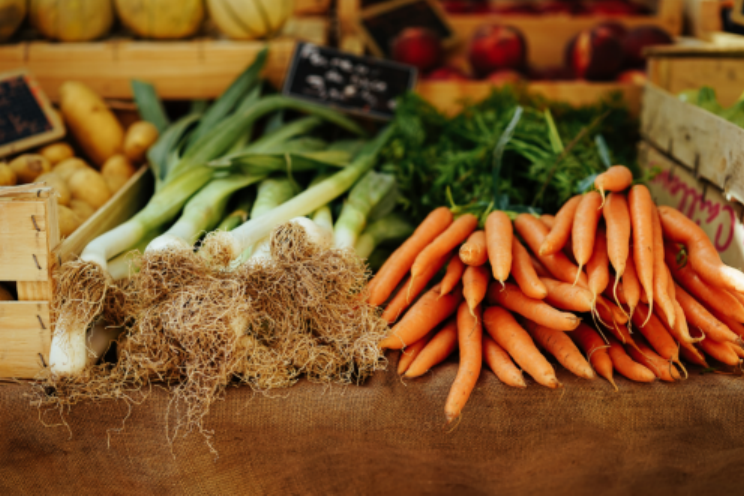}{0.78}{0.5}{0.98}{0.9}
    \end{subfigure}%
    \begin{subfigure}[t]{.25\textwidth}
        \centering
        \drawrectangleonfigure{cyan}{.95\textwidth}{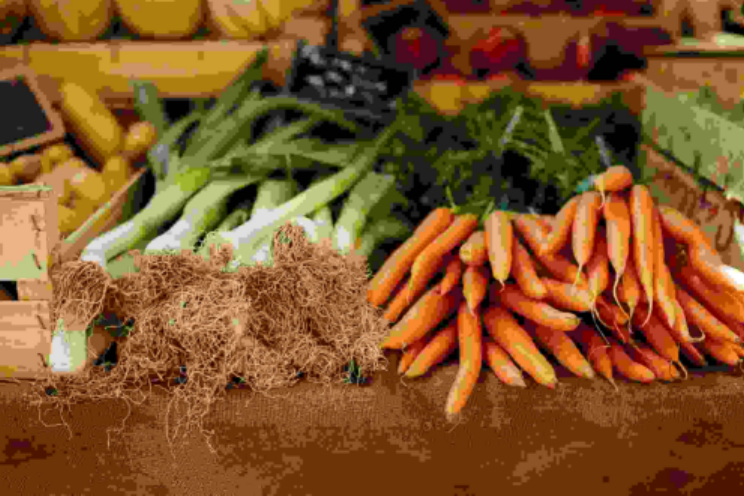}{0.78}{0.5}{0.98}{0.9}
        \caption{\scriptsize\textbf{bpp: 0.18, PSNR: 21.20dB}}
    \end{subfigure}%
    \begin{subfigure}[t]{.25\textwidth}
        \centering
        \drawrectangleonfigure{cyan}{.95\textwidth}{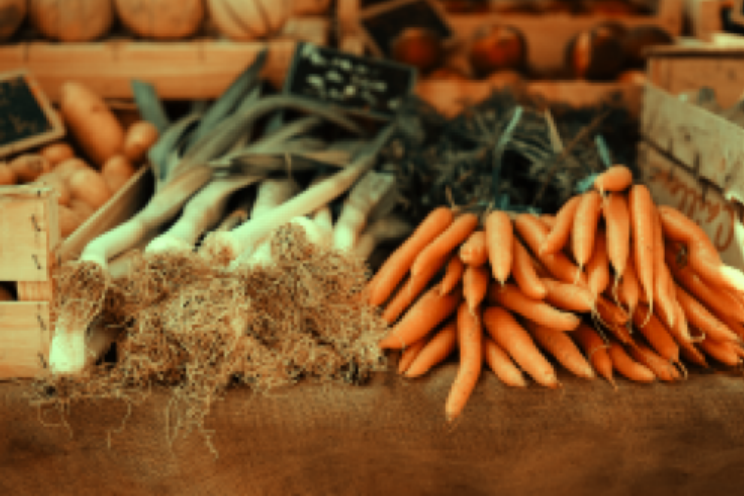}{0.78}{0.5}{0.98}{0.9}
        \caption{\scriptsize\textbf{bpp: 0.19, PSNR: 21.02dB}}
    \end{subfigure}%
    \begin{subfigure}[t]{.25\textwidth}
        \centering
        \drawrectangleonfigure{cyan}{.95\textwidth}{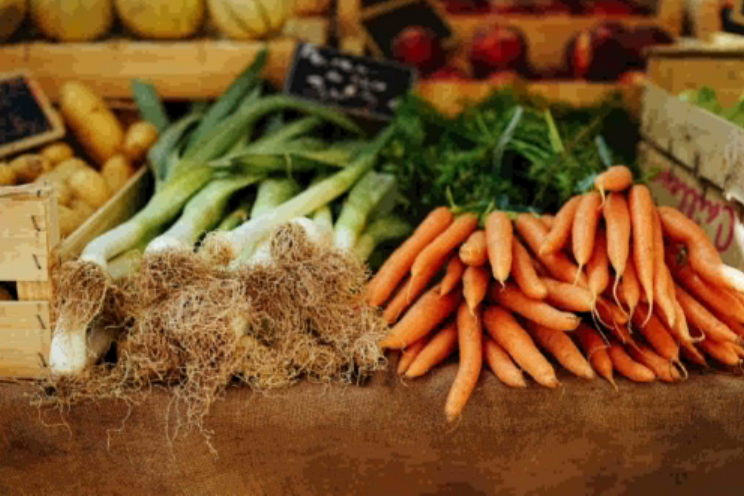}{0.78}{0.5}{0.98}{0.9}
        \caption{\scriptsize\textbf{bpp: 0.17, PSNR: 22.63dB}}
    \end{subfigure}

    \begin{subfigure}[t]{.25\textwidth}
        \centering
        \drawrectangleonfigure{yellow}{.95\textwidth}{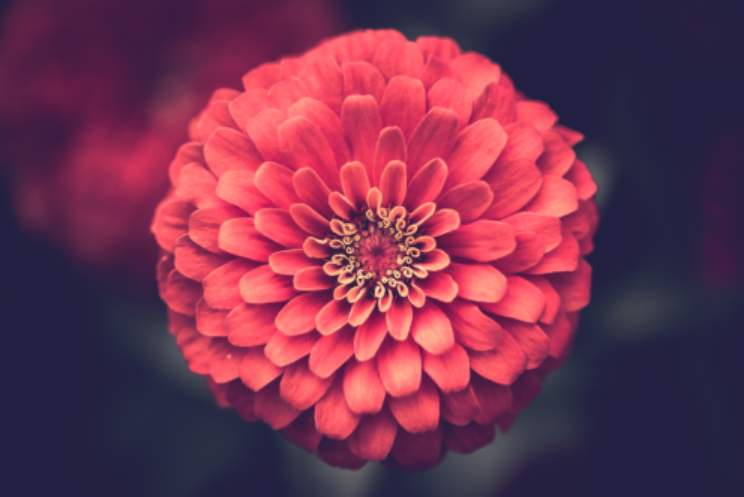}{0.15}{0.58}{0.4}{0.88}
    \end{subfigure}%
    \begin{subfigure}[t]{.25\textwidth}
        \centering
        \drawrectangleonfigure{yellow}{.95\textwidth}{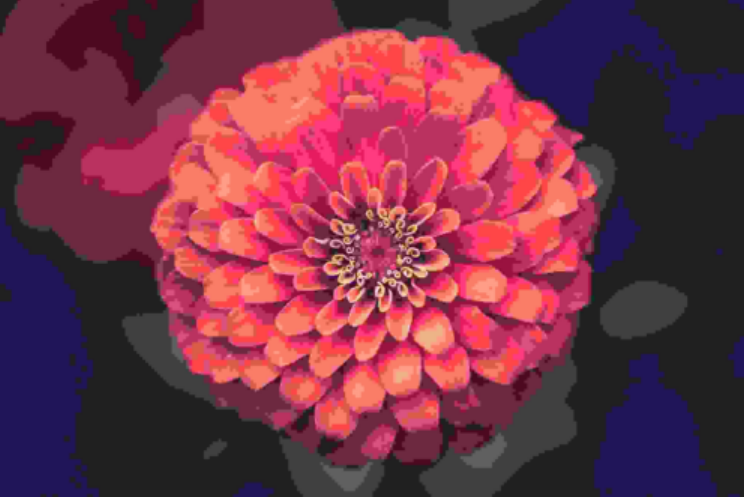}{0.15}{0.58}{0.4}{0.88}
        \caption{\scriptsize\textbf{bpp: 0.14, PSNR: 22.66dB}}
    \end{subfigure}%
    \begin{subfigure}[t]{.25\textwidth}
        \centering
        \drawrectangleonfigure{yellow}{.95\textwidth}{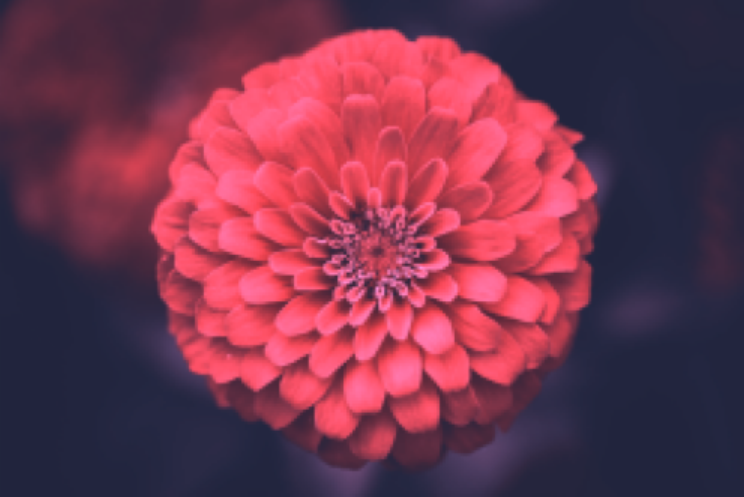}{0.15}{0.58}{0.4}{0.88}
        \caption{\scriptsize\textbf{bpp: 0.12, PSNR: 26.90dB}}
    \end{subfigure}%
    \begin{subfigure}[t]{.25\textwidth}
        \centering
        \drawrectangleonfigure{yellow}{.95\textwidth}{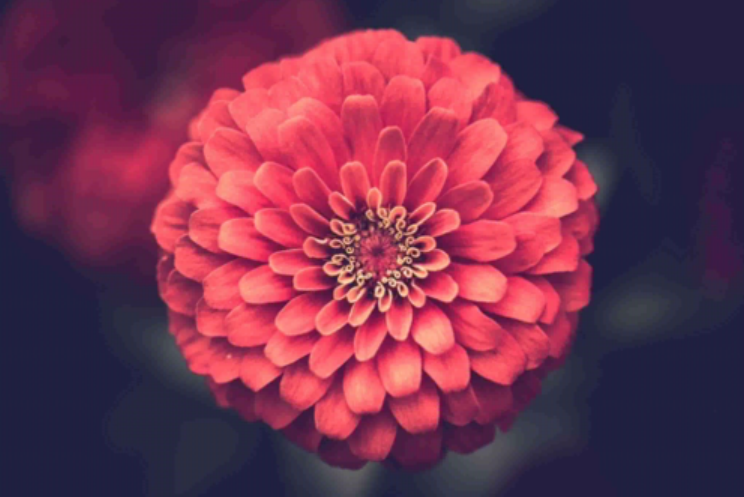}{0.15}{0.58}{0.4}{0.88}
        \caption{\scriptsize\textbf{bpp: 0.12, PSNR: 30.63dB}}
    \end{subfigure}

    \caption{Qualitative performance comparison on example images from the Kodak (top two rows) and the CLIC 2024 (bottom two rows) datasets. Each column shows the original image, JPEG, SVD, and QMF compression results respectively. The bit rate and PSNR values for each compressed image is reported. The colored bounding boxes highlight artifacts produced by JPEG and SVD compression.}
    \label{fig:qualitative_comparison}
\end{figure*}

In the building image (first row), JPEG compression, with a PSNR of 20.22 dB at a bit rate of 0.21 bpp, introduces \emph{blocking artifacts} and changes the facade color, as visible in the blue boxes. SVD compression reduces these artifacts but causes blurriness. Our QMF compression, with a similar bit rate but a higher PSNR (21.93 dB), maintains both texture and sharpness with minimal artifacts.

In the seascape image (second row), JPEG causes blocking and significant \emph{color bleeding artifacts}, such as the redness in the cloud area marked by the red boxes and also on the water surface (outside the red box). SVD reduces color distortion but still has blockiness and blurriness. QMF preserves the color and texture of clouds and water more effectively, resulting in a more visually pleasing image.

In the vegetables image (third row), JPEG yields visible \emph{color distortion} (marked by the cyan boxes), while SVD introduces significant blurriness. QMF, however, effectively preserves the color fidelity and detail.

In the flower image (fourth row), JPEG compression, with a PSNR of 20.22 dB at a bit rate of 0.14 bpp, exhibits severe \emph{color banding artifacts} around the flower boundary. SVD compression offers smoother gradients but remains blurry. Our QMF compression maintains the gradient fidelity and intricate petal distinctions, achieving a significantly higher PSNR of 30.63 dB at a lower bit rate of 0.12 bpp.

\subsection{Run Time} \label{sec:run_time}

The decoding times at bit rates of 0.15 bpp and 0.25 bpp for each method on Kodak and CLIC 2024 are reported in Table \ref{tab:run_time}. All experiments in this section were conducted on 2 Xeon Gold 6140 CPUs @ 2.3 GHz (Skylake), each with 18 cores, and with 192 GiB RAM.

\begin{table}[!t]
    \caption{Mean decoding CPU times for different compression methods at bit rates of 0.15 bpp and 0.25 bpp, measured on the Kodak and CLIC 2024 datasets.}
    \label{tab:run_time}
    \centering
    \def\arraystretch{1.2}
    \resizebox{\linewidth}{!}{
        \begin{tabular}{l | c | c | c | c}
            \toprule
            \multirow{2}{*}{Method} & \multicolumn{2}{c|}{Kodak}              & \multicolumn{2}{c}{CLIC 2024}                                                                                     \\
            \cmidrule(lr){2-3} \cmidrule(lr){4-5}
                                    & Bit rate = 0.15 bpp                      & Bit rate = 0.25 bpp                      & Bit rate = 0.15 bpp & Bit rate = 0.25 bpp \\
            \midrule
            JPEG                    & 4.54 ms                                 & 4.23 ms                                 & 26.76 ms           & 25.75 ms           \\
            SVD                     & 1.33 ms                                 & 1.23 ms                                 & 5.29 ms            & 4.82 ms            \\
            QMF                     & 2.82 ms                                 & 2.66 ms                                 & 9.91 ms            & 9.06 ms            \\
            \bottomrule
        \end{tabular}
    }
\end{table}

QMF and SVD have a significant advantage in decoding speed over JPEG, with SVD being the fastest. Specifically, QMF decodes more than twice as fast as JPEG on the CLIC 2024 dataset across all bit rates. This is due to the heavier FFT operation in the JPEG decoder compared to the lighter matrix multiplication in the QMF decoder. Overall, QMF is preferable for applications requiring high image quality at low bit rates, especially in scenarios where compressed images are frequently accessed or displayed while encoding occurs less often. Examples include web browsing, image hosting, mobile applications, satellite imagery for maps, and interactive gaming applications. These use cases often involve large numbers of images or thumbnails, where faster decoding times are essential for seamless user experiences, even if encoding takes more time or demands higher computational resources.

\subsection{ImageNet Classification Performance} \label{sec:imagenet_classification_performance}

It is relevant to assess the ability of different compression methods in preserving the visual semantic information in images. To this end, we investigate the performance of an image classifier on images compressed using various compression methods. This is particularly crucial in scenarios where the ultimate goal is a vision task such as image classification, rather than maintaining perceived image quality, and we compress images before classification to minimize resource requirements, such as memory and communication bandwidth.

In this experiment, we employed a ResNet-50 classifier \cite{he2016deep}, pre-trained on the original ImageNet \cite{deng2009imagenet} dataset, to classify compressed images from the ImageNet validation set using different compression methods. The classification performance comparison is presented in Figure \ref{fig:imagenet_classification}. Notably, the results indicate that QMF compression achieves over a 5\% improvement in top-1 accuracy compared to JPEG at bit rates under 0.25 bpp and reaches a top-5 accuracy exceeding 70\% at a bit rate of 0.2 bpp. QMF compression leads to higher classification accuracies than JPEG at bit rates up to approximately 0.30 bpp. \remove{Compared to the rate-distortion results (Figure \ref{fig:rate_distortion}), where the take-over point was around 0.25 bpp, this shift suggests that the superiority of QMF over JPEG in preserving visual semantics is even greater than its advantage in maintaining image quality.}

\begin{figure*}[!t]
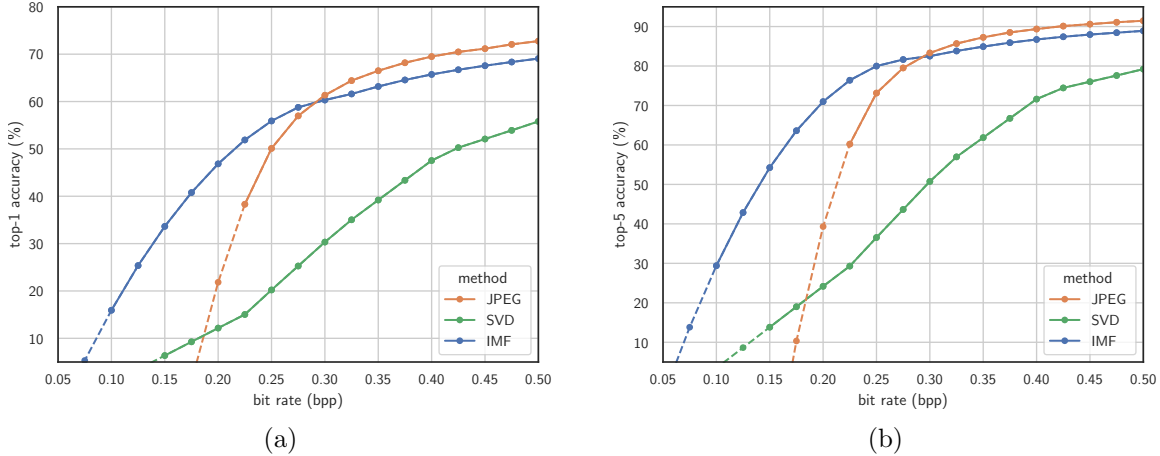

    \centering
    \begin{subfigure}[t]{.5\textwidth}
        \centering
        \resizebox{.9\textwidth}{!}{\input{figures/imagenet_top-1_accuracy.pgf}}
        \caption{}
        \label{fig:imagenet_top1}
    \end{subfigure}%
    \begin{subfigure}[t]{.5\textwidth}
        \centering
        \resizebox{.9\textwidth}{!}{\input{figures/imagenet_top-5_accuracy.pgf}}
        \caption{}
        \label{fig:imagenet_top5}
    \end{subfigure}
    \caption{Impact of different compression methods on ImageNet classification accuracy. A ResNet-50 classifier pre-trained on the original ImageNet images is evaluated using validation images compressed by different methods. Panels (a) and (b) show top-1 and top-5 accuracy plotted as a function of bit rate, respectively. Dashed lines indicate extrapolated values predicted using LOESS \cite{cleveland1988locally} for extremely low bit rates that are otherwise unattainable.}
    \label{fig:imagenet_classification}
\end{figure*}

\subsection{Ablation Studies} \label{sec:ablation_studies}

We conducted ablation studies to investigate the impact of factor bounds, the number of BCD iterations, and patch size on the compression performance of our QMF method. All experiments in this section were performed using the Kodak dataset. We followed the QMF configuration described in Section \ref{sec:setup} and varied only the parameters under ablation one at a time.

\paragraph{Factor Bounds}
Figure \ref{fig:bounds_ablation_psnr} shows the average PSNR as a function of bit rate for QMF using various factor bounds $[\alpha, \beta]$ in Algorithm \ref{alg:bcd_for_qmf}. The results indicate that the interval $[-16, 15]$ yields the optimal performance, showing moderate improvement over both $[-32, 31]$ and $[-128, 127]$, while significantly outperforming $[-8, 7]$. In fact, constraining the factor elements within a sufficiently narrow range can reduce the bit allocation needed, thereby leading to higher compression ratios. Note that in all these cases, the factor elements are represented as the \texttt{int8} data type.

\remove{Narrowing the bounds reduces entropy, which improves the effectiveness of lossless compression in the final stage of our framework and thus reduces the bit rate. However, this more constrained feasible set and reduction in entropy come at the cost of higher reconstruction error. Our experiments showed that using the bounds [-16, 15] provides the best trade-off between entropy and reconstruction quality.}

\paragraph{BCD Iterations}
The next parameter studied is the number of BCD iterations $K$ in Algorithm \ref{alg:bcd_for_qmf}, where each BCD iteration involves one complete cycle of updates across all the columns of both factors. Figure \ref{fig:bounds_ablation_psnr} shows the average PSNR plotted against the bit rate for QMF with different numbers of iterations $K \in \{0, 1, 2, 5, 10\}$. As expected, more iterations consistently resulted in higher PSNR for QMF compression. Without any BCD iterations ($K=0$) and relying solely on the SVD-based initialization given by \eqref{eq:initialization_u} and \eqref{eq:initialization_v}, the results became very poor. However, performance improved significantly after a few iterations, with more than $K=5$ iterations yielding only marginal gains. We found that $K=10$ iterations are sufficient in practice for image compression applications. This makes QMF computationally efficient, as decent compression performance can be achieved even with a limited number of BCD iterations.

\paragraph{Patchification}
Figure \ref{fig:patchsize_ablation_psnr} explores the impact of different patch sizes on QMF performance in terms of PSNR. As observed, a patch size of $(8, 8)$ yields the best performance. A patch size of $(16, 16)$ follows closely, with only marginally lower PSNR at higher bit rates. Conversely, larger patch sizes like $(32, 32)$ or omitting the patchification step altogether significantly degrade compression performance. \remove{Therefore, optimal patchification positively affects performance by effectively modeling the locality and spatial dependencies of neighboring pixels.}

\begin{figure*}[!t]
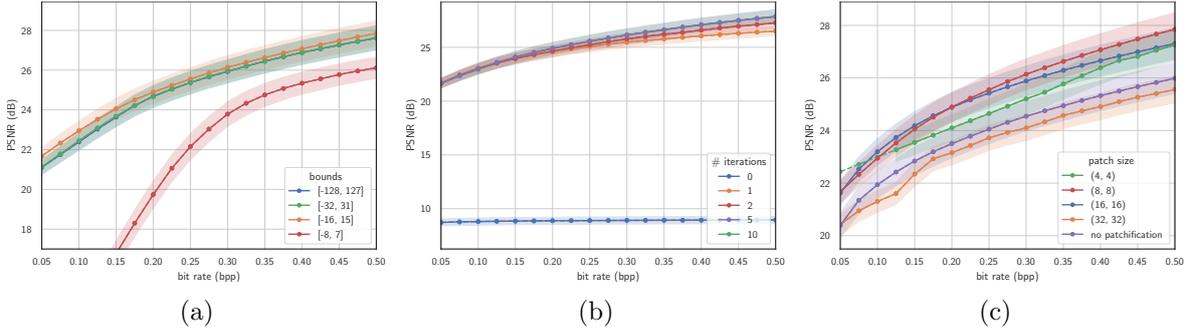

    \centering
    \begin{subfigure}[t]{.33\textwidth}
        \centering
        \resizebox{.95\textwidth}{!}{\input{figures/ablation_bounds_psnr.pgf}}
        \caption{}
        \label{fig:bounds_ablation_psnr}
    \end{subfigure}%
    \begin{subfigure}[t]{.33\textwidth}
        \centering
        \resizebox{.95\textwidth}{!}{\input{figures/ablation_numiters_psnr.pgf}}
        \caption{}
        \label{fig:iters_ablation_psnr}
    \end{subfigure}%
    \begin{subfigure}[t]{.33\textwidth}
        \centering
        \resizebox{.95\textwidth}{!}{\input{figures/ablation_patchsize_psnr.pgf}}
        \caption{}
        \label{fig:patchsize_ablation_psnr}
    \end{subfigure}
    \caption{Ablation studies for QMF. The average PSNR on the Kodak dataset is plotted as a function of bit rate under various experimental conditions: (a) varying the bounds $[\alpha, \beta]$ for the elements of the factor matrices, (b) changing the number of BCD iterations, and (c) adjusting the patch size.}
    \label{fig:ablation_studies}
\end{figure*}

\section{Discussion} \label{sec:discussion}

% Strength: QMF vs SVD and JPEG
All our comparative results (Figure \ref{fig:rate_distortion} and Figure \ref{fig:imagenet_classification}) consistently show that our QMF method outperforms JPEG in both maintaining image quality and preserving visual semantics at low bit rates and remains comparable at higher bit rates. Moreover, QMF consistently demonstrates superior performance compared to SVD across all bit rates. This superiority can be attributed to the integration of quantization with low-rank approximation in QMF, which enables more accurate reconstruction. In contrast, the high sensitivity of SVD to quantization errors, which arise during a separate quantization step, significantly degrades the reconstruction quality.

% Strength: QMF visual quality vs visual semantic
\remove{Although the rate-distortion curves (Figure \ref{fig:rate_distortion}) show that QMF is surpassed by JPEG at a bit rate of around 0.25 bpp, the ImageNet results (Figure \ref{fig:imagenet_classification}) reveal a take-over point at around 0.3 bpp. This small shift suggests that the superiority of QMF over JPEG at low bit rates in terms of preserving visual semantics is even greater than its advantage in maintaining image quality.}

% Limitation 1
As observed in Figure \ref{fig:bounds_ablation_psnr}, contracting the QMF factor bounds from $[-128, 127]$ to $[-16, 15]$ consistently improves the rate-distortion performance. Generally, narrowing the factor bounds $[\alpha, \beta]$ can potentially lower the entropy, thereby improving the effectiveness of lossless compression in the final stage of our framework and subsequently reducing the bit rate. However, this reduction in entropy comes at the cost of increased reconstruction error, as the feasible set in \eqref{eq:qmf_problem} becomes more constrained. This trade-off between entropy and reconstruction quality limits the compression performance of QMF. Therefore, it would be beneficial\remove{to relax the factor bounds to some extent}\remove{to contract the factor bounds less tight} to moderately expand the factor bounds $[\alpha, \beta]$ while simultaneously controlling the entropy of the elements in the factor matrices. We plan to address this in the future by incorporating an entropy-aware regularization term into the current QMF objective function.

% Limitation 2
Patchification with an appropriate patch size (e.g., $(8, 8)$) helps capture local spatial dependencies and, as confirmed by our results in Figure \ref{fig:patchsize_ablation_psnr}, positively impacts the performance of QMF and SVD. However, discontinuities at patch boundaries can introduce \emph{blocking artifacts}, similar to JPEG compression at very low bit rates (see the building image example in Figure \ref{fig:qualitative_comparison}). Moreover, while JPEG suffers more from \emph{color distortion} (e.g., \emph{color bleeding} and \emph{color banding}) at low bit rates, QMF and SVD are more affected by \emph{blurriness}, as observed in the seascape image example in Figure \ref{fig:qualitative_comparison}. As a potential solution for future work, a deep neural network could be trained to remove these artifacts and then integrated as a post-processing module to further enhance the quality of QMF-compressed images.

\section{Conclusion} \label{sec:conclusion}

This work presents a novel lossy image compression method based on quantization-aware matrix factorization (QMF). By representing image data as the product of two smaller factor matrices with elements constrained to bounded integer values, the proposed QMF approach effectively integrates quantization with low-rank approximation. In contrast, traditional compression methods such as JPEG and SVD consider quantization as a separate step, where quantization errors cannot be incorporated into the compression process. The reshaped factor matrices in QMF are compatible with existing lossless compression standards, enhancing the overall flexibility and efficiency of our method. Our proposed iterative algorithm, utilizing a block coordinate descent scheme, has proven to be both efficient and convergent. Experimental results demonstrate that the QMF method significantly outperforms JPEG in terms of PSNR and SIMM at low bit rates and maintains better visual semantic information. This advantage underscores the potential of QMF to set a new standard in lossy image compression, bridging the gap between factorization and quantization.

\remove{A limitation of QMF lies in its inability to control the entropy of the elements in the factor matrices, which could enhance the performance of the subsequent entropy-based lossless coder. We plan to address this in the future by incorporating entropy-aware regularization into the current QMF objective function.}
\section*{Acknowledgments}

This research received funding from Flanders AI Research Program. Sabine Van Huffel and Pooya Ashtari are affiliated to Leuven.AI - KU Leuven institute for AI, B-3000, Leuven, Belgium.

\appendix
\section{Proof of Theorem \ref{the:bcd_subproblem}} \label{app:monotonicity_proof}

We start by proving the closed-form solution \eqref{eq:bcd_closed_form_subproblem_u}, noting that the proof for \eqref{eq:bcd_closed_form_subproblem_v} follows the same reasoning.
The objective function in the subproblem \eqref{eq:bcd_subproblem_u} can be reformulated as follows:
\begin{align}
    \argmin_{\bm{u}_r \in \Z_{[\alpha,\beta]}^M} \|\bm{E}_r - \bm{u}_r \bm{v}_r^\mathsf{T}\|_{\rm F}
    = \argmin_{\bm{u}_r \in \Z_{[\alpha,\beta]}^M}\sum_{i=1}^M \sum_{j=1}^N (e_{ij}^r - u_i^r v_j^r)^2,
    \label{eq:bcd_subproblem_u_reformed}
\end{align}
where $e_{ij}^r$ denotes the element of matrix $\bm{E}_r$ in the $i$th row and $j$th column, and $u_i^r$ and $v_j^r$ are the $i$th and $j$th elements of vectors $\bm{u}_r$ and $\bm{v}_r$, respectively. Since the elements of $\bm{E}_r$ and $\bm{v}_r$ are fixed in problem \eqref{eq:bcd_subproblem_u}, the optimization \eqref{eq:bcd_subproblem_u_reformed} can be decoupled into $M$ optimizations as follows:
\begin{align}
    \argmin_{u^r_i \in \Z_{[\alpha,\beta]}} ~q_i(u_i^r), \quad
    \textstyle \text{where} ~~              q_i(u_i^r) \triangleq \sum_{j=1}^N (e_{ij}^r - u_i^r v_j^r)^2, \quad \forall i\in\{1,\dots,M\}.
    \label{eq:bcd_subproblem_u_decoupled}
\end{align}
The objective functions $q_i(u_i^r)$ in \eqref{eq:bcd_subproblem_u_decoupled} are single-variable quadratic problems. Hence, the global optimum in each decoupled optimization problem can be achieved by finding the minimum of each quadratic problem and then projecting it onto the set $\Z_{[\alpha,\beta]}$. The minimum of each quadratic function in \eqref{eq:bcd_subproblem_u_decoupled}, denoted by $\bar{u}_i^r$, can be simply found by
\begin{align}
    \textstyle \nabla_{u_i^r} q_i(u_i^r) = 0 \implies \bar{u}_i^r = \nicefrac{\sum_{j=1}^N e_{ij}^r v_j^r}{\sum_{j=1}^N v_j^{r^2}},
    \label{eq:bcd_subproblem_u_minimum}
\end{align}
where $\nabla_x$ is the partial derivative with respect to $x$.
Since $q_i$ has a constant curvature (second derivative) and $q_i(\bar{u}_i^r + d)$ is nondecreasing with increasing $|d|$, the value in the set $\Z_{[\alpha,\beta]}$ which is closest to $\bar{u}_i^r$ is the global minimizer of \eqref{eq:bcd_subproblem_u_decoupled}. This value can be reached by projecting $\bar{u}_i^r$ onto the set $\Z_{[\alpha,\beta]}$, namely $u^{r^\star}_i = \clamp_{[\alpha,\beta]}(\round(\bar{u}_i^r))$, which is presented for all $i\in\{1,...,M\}$ in a compact form in \eqref{eq:bcd_closed_form_subproblem_u}.

Since $\bm{u}_r^\star \triangleq (u^{r^\star}_1,\dots,u^{r^\star}_M)$ is the global optimum of optimization \eqref{eq:bcd_subproblem_u_reformed}, it is evident that
\begin{equation}
    \|\bm{E}_r - \bm{u}_r^\star \bm{v}_r^\mathsf{T}\|_{\rm F}  \leq \|\bm{E}_r - \bm{u}_r \bm{v}_r^\mathsf{T}\|_{\rm F}.
\end{equation}
This inequality guarantees a nonincreasing cost function over one update of $\bm u_r$. Following the same reasoning for updates of $\bm v_r$ in \eqref{eq:bcd_closed_form_subproblem_v}, it can be concluded that in each update of \eqref{eq:bcd_closed_form_subproblem_u} and \eqref{eq:bcd_closed_form_subproblem_v}, the cost function is nonincreasing. Therefore, the sequential updates over the columns of $\bm U$ and $\bm V$ in Algorithm \ref{alg:bcd_for_qmf} result in a monotonically nonincreasing cost function in \eqref{eq:qmf_problem}.

\section{Proof of Theorem \ref{thm:convergence}} \label{app:convergence_proof}

To study the convergence of the proposed Algorithm \ref{alg:bcd_for_qmf}, we recast the optimization problem \eqref{eq:qmf_problem} into the following equivalent problem:
\begin{align}
    \begin{split}
        \minimize_{U_{:r} \in \R^{M}, V_{:r} \in \R^{N}, \forall r\in\{1,...,R\}} \quad \Psi(\bm U, \bm V)
    \end{split}
    \label{eq:qmf_surrogate}
\end{align}
where
\begin{align*}
    \Psi(\bm U, \bm V) & \triangleq f_0(\bm U, \bm V) + \sum_{r=1}^R f(U_{:r}) + \sum_{r=1}^R g(V_{:r}), \\
    f_0(\bm U, \bm V)  & \triangleq \|\bm X - \bm U \bm V^{\rm T} \|_{\rm F}^2,                          \\
    f(U_{:r})          & \triangleq \delta_{[a,b]}(U_{:r}) + \delta_\Z(U_{:r}),                          \\
    g(V_{:r})          & \triangleq \delta_{[a,b]}(V_{:r}) + \delta_\Z(V_{:r}),
\end{align*}
with $\delta_\mathcal{B}(\cdot)$ as the indicator function of the nonempty set $\mathcal{B}$ where $\delta_\mathcal{B}(\bm x) = 0$ if $\bm x \in \mathcal{B}$ and $\delta_\mathcal{B}(\bm x) = +\infty$, otherwise. By the definition of functions above, it is easy to confirm that the problems \eqref{eq:qmf_problem} and \eqref{eq:qmf_surrogate} are equivalent.

The unconstrained optimization problem \eqref{eq:qmf_surrogate} consists of the sum of a differentiable (smooth) function $f_0$ and nonsmooth functions $f$ and $g$. This problem has been extensively studied in the literature under the class of nonconvex nonsmooth minimization problems.
% One of the common algorithms applied to such a problem class is the well-known forward-backward-splitting (FBS) algorithm \cite{combettes2011proximal,bauschke2017correction}.
In Algorithm \ref{alg:bcd_for_qmf}, the blocks $U_{:r}$ and $V_{:r}$ are updated sequentially following block coordinate descent (BCD) minimization algorithms, also often called Gauss-Seidel updates or alternating optimization \cite{nesterov2012efficiency,attouch2013convergence}.
Hence, in this convergence study, we are interested in algorithms that allow BCD-like updates for the nonconvex nonsmooth problem of \eqref{eq:qmf_surrogate} \cite{beck2013convergence,bolte2014proximal}. Specifically, we focus on the proximal alternating linearized minimization (PALM) algorithm \cite{bolte2014proximal}, to relate its convergence behavior to that of Algorithm \ref{alg:bcd_for_qmf}.
To that end, we show that the updates of Algorithm \ref{alg:bcd_for_qmf} are related to the updates of PALM on the recast problem of \eqref{eq:qmf_surrogate}, and all the assumptions necessary for the convergence of PALM are satisfied by our problem setting.
It is noted that, for the sake of presentation and without loss of generality, in this proof, we assume each of the matrices $\bm U$ and $\bm V$ has only one column ($R=1$); hence, we only have two blocks in the BCD updates. The iterates in PALM and the presented proof can be trivially extended for more than two blocks.

The PALM algorithm can be summarized as follows:
\begin{enumerate}
    \item Initialize $\bm U^{\rm init} \in \R^{M\times R}$, $\bm V^{\rm init} \in \R^{N\times R}$
    \item For each iteration $k=0,1,...$
          \begin{align}
              \begin{split}
                  \mysubnumber~ \bm U^{k+1} & \in \prox^f_{c_k} \left(\bm U^k - \frac{1}{c_k} \nabla_{\bm U} f_0(\bm U^k, \bm V^k)\right),     \\
                  \mysubnumber~ \bm V^{k+1} & \in \prox^g_{d_k} \left(\bm V^k - \frac{1}{d_k} \nabla_{\bm V} f_0(\bm U^{k+1}, \bm V^k)\right),
              \end{split}
              \label{eq:palm_updates}
          \end{align}
\end{enumerate}
where the proximal map for an extended proper lower semicontinuous (nonsmooth) function $\func{\varphi}{\R^n}{(-\infty,+\infty]}$ and $\gamma > 0$ is defined as $\prox^\varphi_\gamma(\bm x) \triangleq \argmin_{\bm w\in\R^n}\left\{\varphi(\bm w) + \frac{\gamma}{2} \|\bm w - \bm x\|^2_2\right\}$. In \eqref{eq:palm_updates}, $c_k > L_1(\bm V^k)$ and $d_k > L_2(\bm U^{k+1})$ where $L_1 > 0$, $L_2 > 0$ are local Lipschitz moduli, defined in the following proposition.

The following proposition investigates the necessary assumptions (cf. \cite[Asm. 1 and Asm. 2]{bolte2014proximal}) for convergence of iterates in \eqref{eq:palm_updates}.
\begin{prop}[Meeting required assumptions]\label{prop:assumptions}
    The assumptions necessary for the convergence of iterates in \eqref{eq:palm_updates} are satisfied by the functions involved in the problem \eqref{eq:qmf_surrogate}, specifically:
    \begin{enumerate}
        \item The indicator functions $\delta_{[a,b]}$ and $\delta_\Z$ are proper and lower semicontinuous functions, so do the functions $f$ and $g$;
        \item For any fixed $\bm V$, the partial gradient $\nabla_{\bm U} f_0(\bm U, \bm V)$ is globally Lipschitz continuous with modulus $L_1(\bm V) = \|\bm V^T \bm V\|_{\rm F}$. Therefore, for all $\bm U_1,\bm U_2 \in \R^{M\times R}$ the following holds
              \begin{equation*}
                  \|\nabla_{\bm U} f_0(\bm U_1, \bm V) - \nabla_{\bm U} f_0(\bm U_2, \bm V)\| \leq L_1(\bm V) \|\bm U_1 - \bm U_2\|,
              \end{equation*}
              where $\|\cdot\|$ denotes the $\ell_2$-norm of the vectorized input with the proper dimension (here, with the input in $\R^{MR\times 1}$).
              The similar Lipschitz continuity is evident for $\nabla_{\bm V} f_0(\bm U, \bm V)$ as well with modulus $L_2(\bm U) = \|\bm U \bm U^T\|_{\rm F}$.
        \item The sequences $\bm U^k$ and $\bm V^k$ are bounded due to the indicator functions $\delta_{[a,b]}$ with bounded $a$ and $b$. Hence the moduli $L_1(\bm V^k)$ and $L_2(\bm U^k)$ are bounded from below and from above for all $k\in\N$.
        \item The function $f_0$ is twice differentiable, hence, its full gradient $\nabla f_0(\bm U,\bm V)$ is Lipschitz continuous on the bounded set $\bm U \in [a,b]^{M\times R}$, $\bm V \in [a,b]^{N\times R}$. Namely, with $M > 0$:
              \begin{align*}
                  \|\big(\nabla_{\bm U} f_0(\bm U_1, \bm V_1) - \nabla_{\bm U}  f_0(\bm U_2, \bm V_2),                             
                  \nabla_{\bm V} f_0(\bm U_1, \bm V_1)                          - \nabla_{\bm V} f_0(\bm U_2, \bm V_2)\big)\|      
                                                                                \leq M \|(\bm U_1 - \bm U_2, \bm V_1 - \bm V_2)\|,
              \end{align*}
              where $(\cdot,\cdot)$ denotes the concatenation of the two arguments.
        \item The sets $[a,b]$ and integer numbers are semi-algebraic; so are their indicator functions. The function $f_0$ is also polynomial, hence it is semi-algebraic. The sum of these functions results in a semi-algebraic function $\Psi$ in \eqref{eq:qmf_surrogate}, hence $\Psi$ is a Kurdyka-Łojasiewicz (KL) function.
    \end{enumerate}
\end{prop}
By Proposition \ref{prop:assumptions}, the optimization problem \eqref{eq:qmf_surrogate} can be solved by the iterates in \eqref{eq:palm_updates}, due to the following proposition:
\begin{prop}[Global convergence \cite{bolte2014proximal}]\label{prop:convergence}
    With the assumptions in proposition \ref{prop:assumptions} being met by the problem \eqref{eq:qmf_surrogate}, let $\seq{(\bm U^k, \bm V^k)}$ be a sequence generated by the iterates in \eqref{eq:palm_updates}. Then the sequence converges to a critical point $(\bm U^\star, \bm V^\star)$ of the problem \eqref{eq:qmf_surrogate}, where $0 \in \partial \Psi(\bm U^\star, \bm V^\star)$, with $\partial$ as the subdifferential of $\Psi$.
\end{prop}

% In the following, we highlight that the iterates in \eqref{eq:palm_updates} can be implemented more simply and more efficiently by Algorithm \ref{alg:bcd_for_qmf} for the problem of image compression. 
It is noted that the so-called \emph{forward} steps $\bm U^k - \frac{1}{c_k} \nabla_{\bm U} f_0(\bm U^k, \bm V^k)$ and $\bm V^k - \frac{1}{d_k} \nabla_{\bm V} f_0(\bm U^{k+1}, \bm V^k)$ in the $\prox$ operators in \eqref{eq:palm_updates} are replaced by the simple closed-form solutions $\nicefrac{\bm{E}_r \bm{v}_r}{\lVert \bm{v}_r \rVert^2}$ and $\nicefrac{\bm{E}_r^\mathsf{T} \bm{u}_r}{\lVert \bm{u}_r \rVert^2}$ in Algorithm \ref{alg:bcd_for_qmf} at steps \ref{alg:step:u_update:1} and \ref{alg:step:v_update:1} (cf. \eqref{eq:bcd_closed_form_subproblem_u} and \eqref{eq:bcd_closed_form_subproblem_v}), respectively. In the case where the iterates \eqref{eq:palm_updates} are extended to multi-block updates, each block represents one column. This is thanks to the special form of the functions $f_0(\cdot, \bm V^k)$ and $f_0(\bm U^{k+1}, \cdot)$ being quadratic functions, each having a global optimal point, which ensures a descent in each forward step.
Furthermore, the proximal operators $\prox^f_{c_k}$ and $\prox^g_{d_k}$ can efficiently be implemented by the operators $\round$ and $\clamp_{[\alpha,\beta]}$ in \eqref{eq:bcd_closed_form_subproblem_u} and \eqref{eq:bcd_closed_form_subproblem_v} (and equivalently in Algorithm \ref{alg:bcd_for_qmf} at steps \ref{alg:step:u_update:2} and \ref{alg:step:v_update:2}). The equivalence of these steps is proven in the following lemma.

\begin{lem}[$\prox$ implementation]
    Consider the operators $\round$ and $\clamp_{[\alpha,\beta]}$ defined in \eqref{eq:bcd_closed_form_subproblem_u} and \eqref{eq:bcd_closed_form_subproblem_v}.
    % Define the following (elementwise) operators $\func{P_{[a,b]}}{\R}{\R}$, and $\func{T_\Z}{\R}{\R}$:
    % \begin{align}
    %     P_{[a,b]}(x) &\triangleq \min\{\max\{a, x\}, b\},\\
    %     T_\Z(x) &\triangleq \max\{\lfloor x \rfloor, \lfloor x+0.5 \rfloor\}. \label{eq:tz}
    % \end{align}
    Then $\prox^f_{c_k}(\bm W) = \round(\clamp_{[\alpha,\beta]}(\bm W))$ and $\prox^g_{d_k}(\bm Z) = \round(\clamp_{[\alpha,\beta]}(\bm Z))$ for any $\bm W\in \R^{M\times R}$, $\bm Z\in\R^{N\times R}$, and $\round(\clamp_{[\alpha,\beta]}(\cdot))$ being an elementwise operator on the input matrices.
\end{lem}
\begin{proof}
    Define the following norms for a given matrix $\bm W \in \R^{M\times R}$:
    \begin{align*}
         \|\bm W\|_{[a,b]}^2 \triangleq \sum_{i,j \mid a \leq \bm W_{ij} \leq b} \bm W_{ij}^2, \quad
        \|\bm W\|_a^2 \triangleq \sum_{i,j \mid \bm W_{ij} < a} \bm W_{ij}^2, \quad                         
         \|\bm W\|_b^2 \triangleq \sum_{i,j \mid \bm W_{ij} > b} \bm W_{ij}^2.
    \end{align*}
    Moreover, note that the $\round$ operator  can be equivalently driven by the following proximal operator:
    \begin{equation}
        \round(\bm W) = \argmin_{\bm U \in \Z^{M\times R}}\{\|\bm U - \bm W\|^2_F\}.
        \label{eq:equivalence_prox_tz}
    \end{equation}
    The proximal operator $\prox^f_{c_k}(\bm W)$ can be rewritten as
    \begin{align*}
        {\rm{pr}} {\rm{ox}}^f_{c_k}(\bm W) &= \argmin_{\bm U \in \R^{M\times R}}\{\delta_{[a,b]}(\bm U) + \delta_{\Z}(\bm U) + \frac{c_k}{2} \|\bm U - \bm W\|^2_F\} \\
                & = \argmin_{\bm U \in \Z_{[a,b]}^{M\times R}}\{\|\bm U - \bm W\|^2_F\}                                                                             \\
                & = \argmin_{\bm U \in \Z_{[a,b]}^{M\times R}}\{\|\bm U - \bm W\|^2_{[a,b]} + \|\bm U - \bm A\|^2_a + \|\bm U - \bm B\|^2_b\}                       \\
                & = \argmin_{\bm U \in \Z^{M\times R}}\{\|\bm U - \bm W\|^2_{[a,b]} + \|\bm U - \bm A\|^2_a + \|\bm U - \bm B\|^2_b\}                               \\
                & = \argmin_{\bm U \in \Z^{M\times R}}\{\|\bm U - \clamp_{[\alpha,\beta]}(\bm W)\|^2_F\}                                                            \\
                & = \round(\clamp_{[\alpha,\beta]}(\bm W)).
    \end{align*}
    The first equality is due to the definition of $\prox$ which is equivalent to the second equality.
    In the third equality the matrices $\bm A\in\R^{M\times R}$ and $\bm B\in\R^{M\times R}$ have elements all equal to $a$ and $b$, respectively.
    The third equality is due to the fact that replacing $\|\bm U - \bm W\|^2_a + \|\bm U - \bm W\|^2_b$ with $\|\bm U - \bm A\|^2_a + \|\bm U - \bm B\|^2_b$ has no effect on the solution of the minimization. The fourth equality is also trivial due to the involved norms in the third equality. The fifth equality can be easily confirmed by the definition of $\clamp_{[\alpha,\beta]}$.
    Finally, in the last equality, \eqref{eq:equivalence_prox_tz} is invoked.
    It is noted that in the implementation, $\round(\clamp_{[\alpha,\beta]}(\cdot)) = \clamp_{[\alpha,\beta]}(\round(\cdot))$ due to the integrality of the bounds $\alpha,\beta \in \Z$. A similar proof can be trivially followed for $\prox^g_{d_k}(\bm Z) = \round(\clamp_{[\alpha,\beta]}(\bm Z))$ as well.
\end{proof}
Now that the equivalence of iterates \eqref{eq:palm_updates} with the simple and closed-form steps in Algorithm \ref{alg:bcd_for_qmf} is fully established, and the assumptions required for the convergence are verified in proposition \ref{prop:assumptions} to be met by problems \eqref{eq:qmf_surrogate} and \eqref{eq:qmf_problem}, proposition \ref{prop:convergence} can be trivially invoked to establish the convergence of Algorithm \ref{alg:bcd_for_qmf} to a locally optimal point of problem \eqref{eq:qmf_problem}.

\bibliographystyle{elsarticle-num-names} 
\bibliography{ref.bib}

\end{document}